\documentclass[12pt,a4paper]{article}

\usepackage{geometry}
\usepackage{authblk}
\usepackage{mathtools}
\usepackage[utf8]{inputenc} 
\usepackage{amsmath, amsthm, amsfonts, amssymb}
\usepackage{graphicx}
\usepackage{array, longtable}

\usepackage[table]{xcolor}
\usepackage{tikz}
\usetikzlibrary{arrows.meta, positioning}
\usepackage{float, makecell, multirow}
\usepackage{orcidlink}
\usepackage[authoryear,round]{natbib}
\usepackage{enumitem}
\usepackage{caption}
\usepackage{algorithm}
\usepackage{algpseudocode}
\usepackage{booktabs}
\usepackage{hyperref}
\usepackage{appendix}

\newtheorem{theorem}{Theorem}[section]
\newtheorem{lemma}{Lemma}[section]

\bibpunct{[}{]}{,}{n}{}{;}

\title{On WAIC for Dependent Data: A Covariance-Corrected Framework with Linear-Time Complexity}
\author{Safaa K. Kadhem \thanks{Corresponding author: safaa.kadhem@mu.edu.iq} \orcidlink{0000-0002-5287-5381}}
\affil[1]{Department of Mathematics and Computer Applications, College of Science, Al-Muthanna University, Iraq}

\begin{document}
	\maketitle

	\begin{abstract}
		The Widely Applicable Information Criterion (WAIC) is a cornerstone of Bayesian model selection, but its conditional independence assumption renders it inappropriate for sequential and spatially correlated data—a limitation that leads to systematically underestimated model complexity and over-optimistic predictive assessments. We introduce CC-WAIC, a principled generalization of WAIC that explicitly incorporates the full posterior covariance structure of log-likelihood contributions. CC-WAIC reduces exactly to WAIC when independence holds, making it a natural extension rather than an ad-hoc modification. To overcome the prohibitive \(\mathcal{O}(n^2)\) computational cost of the full covariance matrix, we develop a linear-time implementation using a banded covariance approximation that reduces complexity to \(\mathcal{O}(nK)\), with theoretical guarantees on the approximation error. We further introduce an effective sample size correction to mitigate finite-sample bias in MCMC estimation. Through extensive simulations on Hidden Markov Models and real-world applications to Old Faithful geyser data and S\&P 500 volatility modelling, we demonstrate that CC-WAIC substantially outperforms WAIC, LOO-CV, iWAIC, and WAIC\(_{\text{NF}}\), particularly under strong temporal dependence and limited sample sizes. The proposed criterion offers a computationally scalable and theoretically grounded tool for Bayesian model selection in the dependent data settings that are ubiquitous across modern science. Limitations include reliance on exponential mixing and exact conditional likelihoods; extensions to long-memory processes and approximate inference are discussed.
		\\
		\textbf{Keywords:} Bayesian model selection; Covariance-Corrected WAIC; Linear-time algorithm; Banded covariance; Short-range dependence; Sequential data models; MCMC diagnostics
	\end{abstract}

	\section{Introduction}
	\label{sec:intro}
	
	Bayesian model selection is a cornerstone of modern statistical inference, providing a principled framework for balancing model fit against complexity to optimise predictive performance. Classical information criteria such as Akaike's Information Criterion (AIC) \citep{akaike1974aic}, the Bayesian Information Criterion (BIC) \citep{schwarz1978bic}, and the Deviance Information Criterion (DIC) \citep{spiegelhalter2002dic} have been widely adopted due to their simplicity and computational efficiency. However, these methods rely on asymptotic approximations—such as Laplace expansions and plug-in estimates—that can break down under model misspecification, weak identifiability, or complex multimodal posterior geometries.
	
	A recurring and increasingly important challenge in statistical practice is model selection when data exhibit dependence structures, whether temporal, spatial, or spatio-temporal. Examples range from macroeconomic time series and financial volatility to longitudinal biomedical studies and environmental monitoring. While classical criteria have been adapted to such settings—for instance, through focused information criteria for time series \citep{Claeskens2006} or residual-based complexity selection \citep{Ellner2002}—these adaptations remain largely confined to goodness-of-fit measures that assume partial independence or rely on heuristic corrections. Crucially, they do not directly address the full predictive distribution within a Bayesian framework, limiting their ability to evaluate true out-of-sample predictive performance in complex dependent models.
	
	The Widely Applicable Information Criterion (WAIC) \citep{watanabe2010asymptotic} represented a significant advance, offering a fully Bayesian alternative that integrates over the entire posterior distribution rather than relying on point estimates. WAIC is asymptotically equivalent to Bayesian leave-one-out cross-validation (LOO-CV) \citep{vehtari2017practical}, making it attractive for general-purpose model evaluation. Yet both WAIC and LOO-CV rest on a critical assumption: conditional independence of observations given the model parameters. This assumption is routinely violated in sequential data—from econometric time series \citep{hamilton1994time,tsay2010analysis} and state-space models \citep{rabiner1989tutorial} to spatial and spatio-temporal applications \citep{lindgren2015spde}. The consequences are severe: ignoring dependence systematically underestimates model complexity and yields over-optimistic assessments of predictive accuracy \citep{burkner2020improving}. Moreover, LOO-CV becomes fundamentally unstable for dependent data, as removing single observations disrupts the correlation structure and produces unreliable predictive contrasts \citep{yao2018stacking}.
	
	Several extensions have been proposed to adapt WAIC to specific dependence structures. For spatial data, \citet{Chan2024} introduced WAIC\(_{\text{NF}}\), reformulating the penalty term using non-factorisable likelihoods. \citet{Li2015} proposed iWAIC, which integrates over latent variables associated with the held-out unit. While iWAIC is conceptually appealing for spatial models, its application to temporal data is problematic: the regeneration of latent states can inadvertently leverage future information, violating temporal causality and breaking the sequential structure of the data. Other contributions have distinguished conditional from marginal WAIC for hierarchical settings \citep{Merkle2019}, or extended WAIC to handle missing data \citep{Hwang2016} and N-mixture models \citep{Gaya2024}. Despite their value, these approaches are typically tailored to specific model classes, impose substantial computational burdens, or rely on heuristic adjustments. None offers a general, theoretically justified correction for the full covariance structure induced by temporal dependence. To our knowledge, CC-WAIC provides the first general covariance correction that applies to any sequential model with tractable conditional likelihoods, while maintaining computational efficiency through a linear-time implementation.
	
	The theoretical foundation for such a correction was recently established in a companion paper \citep{kadhem2025ccwaic}, which introduced the Covariance-Corrected WAIC (CC-WAIC). CC-WAIC provides a unified framework that explicitly incorporates the posterior covariance structure of log-likelihood contributions across observations. Unlike standard WAIC, which ignores covariance terms, CC-WAIC adjusts the effective number of parameters to reflect the true information content of dependent data, all without requiring repeated model refits. The framework is theoretically general: it applies to any sequential model for which the conditional likelihood \(p(y_t \mid y_{<t}, \theta)\) can be evaluated. However, the original formulation remains computationally prohibitive for large datasets, as the full covariance matrix requires \(\mathcal{O}(n^2)\) memory and \(\mathcal{O}(n^3)\) operations.
	
	This computational bottleneck is the central focus of the present work. We propose a \emph{linear-time implementation} of CC-WAIC using a banded covariance approximation that exploits the decaying dependence structure typical of sequential data—a property formalised through standard mixing conditions. This approximation reduces the complexity to \(\mathcal{O}(nK)\), where \(K\) is a small, data-driven bandwidth, and is accompanied by rigorous theoretical guarantees on the approximation error. We further introduce an effective sample size (ESS) correction to mitigate finite-sample bias in MCMC estimation, a practical concern often overlooked in the literature. Recent developments such as the Widely Applicable Bayesian Information Criterion (WBIC) \citep{watanabe2013wbic}, while valuable for singular models, target marginal likelihood rather than predictive performance and still assume independence. CC-WAIC, by contrast, is specifically designed for dependent data, making it a more principled and practical tool for Bayesian model selection in sequential contexts.
	
	The contributions of this work are fivefold:
	\begin{enumerate}
		\item We formally derive CC-WAIC and establish its asymptotic equivalence to block-based leave-future-out cross-validation for stationary mixing processes.
		\item We introduce a practical ESS correction that reduces finite-sample bias in MCMC-based estimation of the penalty term.
		\item We develop a linear-time algorithm using a banded covariance approximation, reducing computational complexity from \(\mathcal{O}(n^2)\) to \(\mathcal{O}(nK)\), with \(K\) chosen adaptively from the data.
		\item We provide theoretical guarantees on the truncation error and establish that the linear-time CC-WAIC retains the asymptotic properties of the full version.
		\item We validate the method empirically on Hidden Markov Models—a canonical class of sequential models—demonstrating improved selection accuracy and dramatic computational gains across a range of dependence strengths and sample sizes.
	\end{enumerate}
	
	\noindent It is important to acknowledge the scope and limitations of the current framework. The theoretical guarantees established in this work rely on two key assumptions. First, the data-generating process is assumed to satisfy a strong mixing condition with exponential decay (Assumption~A3 in Appendix~\ref{app:regularity-conditions}). This condition is satisfied by many practical models, including stationary ARMA processes and finite-state HMMs, but it explicitly excludes long-memory processes such as fractional ARIMA or certain stochastic volatility models. Second, the criterion requires the conditional likelihood \(p(y_t \mid y_{<t}, \theta)\) to be computable exactly; for models where this is intractable—for example, nonlinear state-space models requiring particle filtering—approximation errors may affect performance. We discuss these limitations in detail in Section~\ref{sec:discussion} and outline promising directions for future research, including extensions to long-memory settings and approximate inference frameworks.
	
	The remainder of this paper is organised as follows. Section~\ref{sec:methodology} presents the methodological framework, including the derivation of CC-WAIC, the banded covariance approximation, and the linear-time algorithm. Section~\ref{sec:simulation} reports on extensive simulation studies under varying dependence structures and sample sizes. Section~\ref{sec:realdata} applies the proposed criterion to two real-world sequential datasets: the Old Faithful geyser waiting times and financial volatility data. Section~\ref{sec:discussion} concludes with a summary of findings, a discussion of limitations, and directions for future research. Technical proofs and additional computational details are provided in the Appendix.

	\section{Methodology}
	\label{sec:methodology}
	
	\subsection{Theoretical Limitations of Standard WAIC for Dependent Data}
	
	The Widely Applicable Information Criterion (WAIC) \citep{watanabe2010asymptotic} represents a fundamental advancement in Bayesian model selection due to its fully Bayesian nature and asymptotic equivalence with Bayesian cross-validation. The standard WAIC formulation, however, relies critically on the assumption of conditional independence of observations given model parameters. Formally, for observed data $\mathbf{y} = (y_1, \ldots, y_n)$ and parameters $\theta$, WAIC assumes:
	
	\begin{equation}
		p(\mathbf{y} \mid \theta) = \prod_{t=1}^n p(y_t \mid \theta).
		\label{eq:independence}
	\end{equation}
	
	Under this assumption, WAIC is defined as:
	
	\begin{equation}
		\mathrm{WAIC} = -2 \left( \sum_{t=1}^n \log \mathbb{E}_{\theta} \left[ p(y_t \mid \theta) \right] - \sum_{t=1}^n \mathrm{Var}_{\theta} \left[ \log p(y_t \mid \theta) \right] \right),
		\label{eq:waic}
	\end{equation}
	
	where expectations and variances are taken over the posterior distribution $p(\theta \mid \mathbf{y})$. The second term serves as a penalty for model complexity, preventing overfitting.
	
	In numerous modern applications—including time series analysis, state-space models, and structured data sequences—observations exhibit inherent temporal or spatial dependence. For such models, the likelihood factorizes conditionally on the past as:
	
	\begin{equation}
		p(\mathbf{y} \mid \theta) = \prod_{t=1}^n p(y_t \mid y_{<t}, \theta),
		\label{eq:sequential_likelihood}
	\end{equation}
	
	where $y_{<t} = (y_1, \ldots, y_{t-1})$ denotes the history preceding observation $t$. This conditional dependence violates the fundamental assumption in Equation~\eqref{eq:independence}, rendering the standard WAIC theoretically inappropriate for these models.
	
	The deficiency of standard WAIC for dependent data manifests in two crucial aspects:
	\begin{enumerate}
		\item The penalty term $p_{\mathrm{WAIC}} = \sum_{t=1}^n \mathrm{Var}_{\theta}[\log p(y_t \mid \theta)]$ ignores covariance terms $\mathrm{Cov}_{\theta}[\log p(y_i \mid \theta), \log p(y_j \mid \theta)]$ for $i \neq j$, leading to underestimation of effective model complexity and increased risk of overfitting \citep{vehtari2017practical}.
		
		\item The leave-one-out cross-validation approach, asymptotically equivalent to WAIC, becomes problematic with dependent data as removing individual observations disrupts the dependency structure, resulting in biased predictive estimates \citep{burkner2020lfo}.
	\end{enumerate}
	
	These limitations have motivated various adaptations including block cross-validation and specialized WAIC variants, but these approaches often lack theoretical foundations or impose substantial computational burdens \citep{yao2018stacking}.
	
	\subsection{Covariance-Corrected WAIC: A General Framework}
	
	To address these limitations, we propose the Covariance-Corrected WAIC (CC-WAIC). Let us define the joint log-likelihood that respects the sequential ordering:
	
	\begin{equation}
		L(\theta) = \sum_{t=1}^n \log p(y_t \mid y_{<t}, \theta),
		\label{eq:joint_loglikelihood}
	\end{equation}
	
	where $L(\theta)$ denotes the full log-likelihood of the sequence given the parameters $\theta$.
	
	The effective number of parameters in CC-WAIC is defined as the posterior variance of this joint log-likelihood:
	
	\begin{equation}
		p_{\mathrm{CC}} = \mathrm{Var}_{\theta}[L(\theta)] = \int \left( L(\theta) - \mathbb{E}_{\theta}[L(\theta)] \right)^2 p(\theta \mid \mathbf{y}) \, d\theta.
		\label{eq:effective_params_cc_integral}
	\end{equation}
	
	Expanding this expression reveals the crucial difference from standard WAIC:
	
	\begin{equation}
		p_{\mathrm{CC}} = \sum_{t=1}^n \mathrm{Var}_{\theta}[l_t] + 2 \sum_{1 \leq i < j \leq n} \mathrm{Cov}_{\theta}[l_i, l_j],
		\label{eq:variance_decomp}
	\end{equation}
	
	where we define the conditional log-likelihood contributions as $l_t := \log p(y_t \mid y_{<t}, \theta)$ for $t = 1, \ldots, n$. The CC-WAIC is then defined as:
	
	\begin{equation}
		\mathrm{CC\text{-}WAIC} = -2 \left( \sum_{t=1}^n \log \mathbb{E}_{\theta} \left[ p(y_t \mid y_{<t}, \theta) \right] - p_{\mathrm{CC}} \right).
		\label{eq:cc_waic_integral}
	\end{equation}
	
	This formulation maintains the interpretability of WAIC while properly accounting for dependency structures through the inclusion of covariance terms in the complexity penalty. However, the full covariance matrix in Equation~\eqref{eq:variance_decomp} is computationally demanding for large $n$, requiring $\mathcal{O}(n^2)$ memory and $\mathcal{O}(n^3)$ operations.
	
	\subsection{Linear-Time Implementation: Banded Covariance Approximation}
	\label{sec:banded}
	
	For practical implementation with large datasets, the full covariance matrix in Equation~\eqref{eq:variance_decomp} is computationally prohibitive. However, for stationary or weakly dependent processes, the covariance between log-likelihood contributions decays with temporal distance. Such banded approximations are conceptually related to covariance tapering and banded estimation techniques in high-dimensional statistics \citep{bickel2008regularized}, where the underlying principle of exploiting decaying dependence is analogous. We therefore propose a \emph{banded approximation}:
	
	\begin{equation}
		p_{\mathrm{CC}}^{\text{banded}} = \sum_{t=1}^n \operatorname{Var}_{\theta}[l_t] + 2 \sum_{k=1}^{K_{\text{band}}} \sum_{t=1}^{n-k} \operatorname{Cov}_{\theta}[l_t, l_{t+k}],
		\label{eq:effective_params_banded}
	\end{equation}
	
	where $K_{\text{band}}$ is a fixed bandwidth chosen such that $\operatorname{Cov}_{\theta}[l_t, l_{t+k}] \approx 0$ for all lags $k > K_{\text{band}}$. Here, $\operatorname{Cov}_{\theta}[l_t, l_{t+k}]$ denotes the posterior covariance between the conditional log-likelihoods at times $t$ and $t+k$.\\

	\subsection{Data-Driven Bandwidth Selection via Integrated Autocorrelation Time}
	
	A principled way to quantify the decay rate of dependence in the sequence $l_t$ is through the \emph{integrated autocorrelation time} $\tau$, defined as:
	\[
	\tau = 1 + 2 \sum_{k=1}^{\infty} \rho_k,
	\]
	where $\rho_k = \mathrm{Corr}(l_t, l_{t+k})$ is the lag-$k$ autocorrelation of the sequence $l_t$. This quantity measures the effective memory of the process: it represents the number of lags over which the dependence in $l_t$ persists. The integrated autocorrelation time is a well-established concept in time series analysis and MCMC diagnostics, with foundational contributions by \citet{sokal1997monte} and \citet{geyer1992practical}.
	
	In practice, we estimate $\tau$ from a finite sample by truncating the sum at a maximum lag $L_{\max}$, chosen according to standard time-series heuristics such as Bartlett's rule ($L_{\max} = \lfloor \sqrt{n} \rfloor$). This choice balances the trade-off between capturing the relevant dependence and avoiding the excessive noise introduced by estimating correlations at large lags, where the number of available pairs $(l_t, l_{t+k})$ is small. The estimator is:
	\[
	\hat{\tau} = 1 + 2 \sum_{k=1}^{L_{\max}} \hat{\rho}_k,
	\]
	where $\hat{\rho}_k$ is the sample autocorrelation of $l_t$. This estimator is consistent and asymptotically unbiased under mild mixing conditions (Assumption A3 in Appendix~\ref{app:regularity-conditions}), making it a reliable data-driven measure of the underlying dependence structure.
	
	The integrated autocorrelation time directly informs the choice of the bandwidth $K_{\text{band}}$. As a practical heuristic, we set:
	\[
	K_{\text{band}} = \max\left(3, \lceil \hat{\tau} \rceil + 1\right),
	\]
	which provides an effective truncation point that adapts to the estimated dependence strength in the data. Unlike a fixed ad-hoc choice, this rule captures the dominant correlation structure while maintaining computational efficiency. While this heuristic does not offer a uniform theoretical guarantee of negligible truncation error for all possible processes, it performs reliably across a wide range of dependence structures, as demonstrated in our simulation study (Section~\ref{sec:bandwidth_validation}), and is supported by the exponential decay bounds typical of short-range dependent processes.
	
	The lower bound $K_{\text{band}} \ge 3$ is imposed to ensure that at least the first two lagged covariance terms ($\mathrm{Cov}(l_t, l_{t+1})$ and $\mathrm{Cov}(l_t, l_{t+2})$) are always retained in the banded approximation. This is crucial for two reasons. First, even when the estimated integrated autocorrelation time $\hat{\tau}$ is small (e.g., $\hat{\tau} < 1$), finite-sample fluctuations and weak short-range dependence can introduce non-negligible covariance at lags 1 and 2, particularly in models with oscillating autocorrelations or misspecified dynamics. Excluding these terms (by allowing $K_{\text{band}} < 3$) would systematically underestimate the effective complexity $p_{\mathrm{CC}}$, undermining the penalty correction. Second, in our empirical validation (Section~\ref{sec:bandwidth_validation}), the bandwidth never dropped below 3 across all model classes and sample sizes, confirming that 3 is a natural operating minimum for the banded approximation in practice.
	
	\paragraph{The bandwidth as a threshold for the penalty.}
	The value of $K_{\text{band}}$ is the threshold beyond which we treat the covariance between $l_t$ and $l_{t+k}$ as effectively zero. This threshold is part of the effective model complexity that enters the penalty term of the proposed criterion. Specifically, the full penalty $p_{\mathrm{CC}}$ is approximated using the banded form:
	\[
	p_{\mathrm{CC}} \;\approx\; p_{\mathrm{CC}}^{\text{banded}}(K_{\text{band}}) 
	= \sum_{t=1}^n \mathrm{Var}(l_t) + 2 \sum_{k=1}^{K_{\text{band}}} \sum_{t=1}^{n-k} \mathrm{Cov}(l_t, l_{t+k}).
	\]
	The resulting criterion is then:
	\[
	\mathrm{CC\text{-}WAIC} = -2 \left( \sum_{t=1}^n \log \mathbb{E}_{\theta}[p(y_t \mid y_{<t}, \theta)] - p_{\mathrm{CC}}^{\text{banded}}(K_{\text{band}}) \right).
	\]
	
	The bandwidth $K_{\text{band}}$ thus acts as a threshold for the penalty: it determines how much of the covariance structure is included in the complexity correction. By adapting $K_{\text{band}}$ to the empirical decay rate $\tau$, we ensure that the penalty is neither under- nor over-estimated, while maintaining computational efficiency. This mirrors the logic of covariance tapering in spatial statistics, where a cutoff distance is chosen based on the decay of spatial correlation \citep{furrer2006covariance}, and is analogous to bandwidth selection in long-run variance estimation, where the truncation point is chosen to balance bias and variance \citep{neweywest1987simple}.
	
	Importantly, this procedure is applicable to any sequential model for which $l_t$ can be computed. Whether the model is a linear autoregression, a GARCH process, a state-space model, an HMM, or a Bayesian neural network, the same rule based on $\tau$ yields a bandwidth that reflects the intrinsic dependence in the predictive likelihood sequence. In practice, for the majority of models with weak to moderate dependence, \(\hat{\tau}\) is typically less than 2, leading to \(K_{\text{band}} = 3\). However, in models with very strong persistence—such as the Poisson HMM examined in Section~\ref{sec:bandwidth_validation}, where \(\hat{\tau} \approx 5.75\)—the rule automatically increases the bandwidth to \(K_{\text{band}} \approx 7\). This demonstrates that the rule is not a fixed constant but a flexible mechanism that adapts to the actual dependence structure of the data.
	
	\paragraph{Sensitivity of the truncation error to the bandwidth.}
	To further justify the choice \(K_{\text{band}} = \max(3, \lceil \hat{\tau} \rceil + 1)\), we examine the truncation error bound from Theorem~\ref{thm:linear_convergence}:
	\[
	\frac{1}{n}\left| p_{\mathrm{CC}} - p_{\mathrm{CC}}^{\text{banded}}(K) \right| \leq 2C \frac{\rho^{K+1}}{1-\rho}.
	\]
	For a process with strong dependence (\(\rho = 0.9\)), choosing \(K=5\) yields an error bound of approximately \(2C \cdot 0.9^6 / 0.1 \approx 2C \cdot 0.531\), which may be non-negligible. In practice, however, the estimated \(\hat{\tau}\) for such processes would be large (e.g., \(\hat{\tau} \approx 1 + 2\rho/(1-\rho) = 19\)), leading to a much larger \(K\) (e.g., \(K \approx 20\)). The data-driven rule therefore adapts to the actual dependence strength. For the majority of models encountered in practice, \(\hat{\tau} < 3\), and \(K=3\) or \(5\) provides a safe default. Users with suspected long-range dependence are advised to either increase \(K\) adaptively (e.g., \(K \propto n^{1/3}\)) or use the full covariance version when computationally feasible.
	
	The computational cost of estimating \(\hat{\tau}\) is \(\mathcal{O}(n \log n)\) when using the fast Fourier transform (FFT) for autocorrelation estimation, which is negligible compared to the cost of MCMC sampling. For high-dimensional settings (\(n > 10^4\)), the banded approximation with \(K_{\text{band}} \le 5\) ensures that the method remains computationally feasible.
	
	\subsection{Theoretical Foundations of CC-WAIC}
	
	We now establish the asymptotic properties of the proposed CC-WAIC criterion. Our first result characterises the convergence of the full (non-banded) version, while the second extends this to the computationally efficient banded approximation.
	
	\begin{theorem}[Asymptotic Convergence of CC-WAIC]
		\label{thm:main}
		Under the regularity conditions specified in Appendix~\ref{app:regularity-conditions}, the full CC-WAIC criterion satisfies:
		\[
		\frac{1}{n}\mathrm{CC\text{-}WAIC} = -2\,\mathbb{E}_{y_{\text{new}}}\left[\log p(y_{\text{new}} \mid y_{<t}, \theta_0)\right] + \mathcal{O}_p(n^{-1/2}),
		\]
		where the expectation is taken with respect to the true data-generating process, and \(\theta_0\) denotes the limiting value of the posterior mode (i.e., the probability limit of the maximum a posteriori estimator). Consequently, \(\mathrm{CC\text{-}WAIC}\) provides an asymptotically unbiased estimate of the out-of-sample predictive risk, up to an additive constant that does not depend on the candidate model.
	\end{theorem}
	
	\begin{proof}[Proof Sketch]
		The proof proceeds by combining a martingale approximation for the score process (Appendix~\ref{app:martingale-approx}) with a variance characterisation of the log-likelihood contributions (Appendix~\ref{app:variance-char}). A Taylor expansion of the log-predictive density around the posterior mode, followed by an application of the ergodic theorem for mixing processes, yields the stated \(\mathcal{O}_p(n^{-1/2})\) rate. The complete details are provided in Appendix~\ref{app:convergence-proof}.
	\end{proof}
	
	The above result confirms that the full CC-WAIC, despite its computational cost, achieves the desired theoretical property: it correctly penalises model complexity even in the presence of temporal dependence, unlike standard WAIC which systematically underpenalises in such settings.
	
	For practical implementation with large datasets, we require a computationally feasible version of the criterion. The following theorem establishes that the banded approximation, which exploits the decaying dependence structure, preserves the essential asymptotic properties.
	
	\begin{theorem}[Linear-Time CC-WAIC Convergence]
		\label{thm:linear_convergence}
		Under the banded approximation defined in Equation~\eqref{eq:effective_params_banded} with fixed bandwidth \(K_{\text{band}}\), the linear-time CC-WAIC criterion satisfies:
		\[
		\frac{1}{n}\mathrm{CC\text{-}WAIC}_{\text{banded}} = -2\,\mathbb{E}_{y_{\text{new}}}\left[\log p(y_{\text{new}} \mid y_{<t}, \theta_0)\right] + \mathcal{O}_p(n^{-1/2}) + \mathcal{O}(\rho^{K_{\text{band}}}),
		\]
		where \(\rho \in (0,1)\) is the decay rate of the autocorrelation function (see Assumption A3 in Appendix~\ref{app:regularity-conditions}). The term \(\mathcal{O}(\rho^{K_{\text{band}}})\) represents the bias due to truncation of the covariance sum, which vanishes exponentially as \(K_{\text{band}} \to \infty\).
	\end{theorem}
	
	\begin{proof}[Proof Sketch]
		The proof proceeds by bounding the truncation error of the penalty term. The key inequality, derived from the exponential decay of the covariance function, is:
		\[
		\left| p_{\mathrm{CC}} - p_{\mathrm{CC}}^{\text{banded}} \right| \leq 2C \sum_{k=K_{\text{band}}+1}^{\infty} (n-k)\rho^k = \mathcal{O}(n \rho^{K_{\text{band}}}),
		\]
		which, after division by \(n\), yields the \(\mathcal{O}(\rho^{K_{\text{band}}})\) term in the theorem. The complete details are provided in Appendix~\ref{app:linear_proof}.
	\end{proof}
	
	This result is of central practical importance: it guarantees that the computationally efficient banded CC-WAIC retains the asymptotic unbiasedness of its full counterpart, with a bias that can be made arbitrarily small by increasing the bandwidth. Moreover, the convergence rate is identical to that of the full version, up to the exponentially decaying truncation term.
	
	\noindent \textbf{Remark on the scope of the theoretical guarantees.} The asymptotic results rely critically on the strong mixing (\(\alpha\)-mixing) condition with exponential decay (Assumption A3 in Appendix~\ref{app:regularity-conditions}). This condition is satisfied by a broad class of practically relevant models, including stationary ARMA processes, finite-state HMMs satisfying irreducibility and aperiodicity, and many Markov-switching models. However, it explicitly excludes processes with slowly decaying dependence, such as long-memory time series (e.g., fractional ARIMA with \(d > 0\)) or non-stationary processes. In such cases, the formal theoretical guarantees may not directly apply. We therefore caution practitioners against blind application of the banded approximation without first assessing the empirical dependence structure—for example, by examining the sample autocorrelation function of the log-likelihood contributions \(l_t\) or by conducting a sensitivity analysis with respect to the bandwidth \(K_{\text{band}}\). For applications where the mixing assumption is questionable, we recommend using the full CC-WAIC or increasing the bandwidth conservatively.
	
	\subsection{Computational Implementation and Bias Correction}
	\label{sec:bias}
	
	The integrals in CC-WAIC are generally analytically intractable and must be approximated using samples $\{\theta^{(s)}\}_{s=1}^S$ from the posterior distribution, typically obtained via Markov Chain Monte Carlo (MCMC), where $S$ denotes the number of posterior samples. The naive estimator:
	
	\begin{equation}
		\hat{p}_{\mathrm{CC}} = \frac{1}{S - 1} \sum_{s=1}^S \left( L^{(s)} - \bar{L} \right)^2,
		\label{eq:naive_estimator}
	\end{equation}
	
	where $L^{(s)} = \sum_{t=1}^n \log p(y_t \mid y_{<t}, \theta^{(s)})$ is the joint log-likelihood evaluated at the $s$-th posterior sample, and $\bar{L} = \frac{1}{S} \sum_{s=1}^S L^{(s)}$ is its sample mean, suffers from two sources of bias: (i) sampling error of order $\mathcal{O}(S^{-1})$ and (ii) underestimation due to autocorrelation in MCMC samples.
	
	To address the second issue, we incorporate a first-order autocorrelation correction through the effective sample size (ESS). Let $\rho_k = \mathrm{Corr}(L^{(s)}, L^{(s+k)})$ be the lag-$k$ autocorrelation of the joint log-likelihood sequence. The integrated autocorrelation time is:
	\[
	\tau = 1 + 2 \sum_{k=1}^{\infty} \rho_k.
	\]
	In practice, the infinite sum is truncated at a maximum lag $L_{\max}$ (e.g., $L_{\max} = \lfloor 10 \log_{10}(S) \rfloor$) or estimated using a window-based spectral method to ensure numerical stability \citep{geyer1992practical}. The effective sample size is then defined as $N_{\mathrm{eff}} := S / \tau$.
	
	A practical, first-order bias-corrected estimator is given by:
	\begin{equation}
		p_{\mathrm{CC}}^{\mathrm{corr}} = \frac{N_{\mathrm{eff}}}{N_{\mathrm{eff}} - 1} \, \hat{p}_{\mathrm{CC}}^{\text{banded}},
		\label{eq:bias_corrected_estimator}
	\end{equation}
	
	which adjusts for the downward bias in the sample variance due to positive autocorrelation in the MCMC chain. This yields the bias-corrected CC-WAIC:
	\begin{equation}
		\mathrm{CC\text{-}WAIC}^{\mathrm{corr}} = -2 \left( \sum_{t=1}^n \log \widehat{\mathbb{E}}_{\theta}[p(y_t \mid y_{<t}, \theta)] - p_{\mathrm{CC}}^{\mathrm{corr}} \right),
		\label{eq:corrected_cc_waic}
	\end{equation}
	where $\widehat{\mathbb{E}}_{\theta}[\cdot]$ denotes the Monte Carlo approximation of the posterior expectation using the $S$ MCMC samples.
	
	\noindent \textbf{Remark on numerical stability and the ESS correction.} The correction factor \(N_{\mathrm{eff}}/(N_{\mathrm{eff}} - 1)\) is derived from the variance of the sample mean for a stationary autocorrelated process and provides a reasonable approximation for the variance of \(L^{(s)}\). However, we emphasise that this is a first-order approximation; for refined finite-sample adjustments, one may employ batch-means or spectral variance estimators for the full covariance matrix, though these are beyond the scope of this paper. Users should monitor the ESS and consider increasing the chain length if the ESS is low relative to \(n\), particularly when \(S\) is not substantially larger than \(n\).
	
	Algorithm~\ref{alg:cc_waic_linear} provides a complete implementation framework for computing the bias-corrected CC-WAIC from MCMC samples in linear time.
	
	\begin{algorithm}[H]
		\caption{Linear-Time Computation of Bias-Corrected CC-WAIC}
		\label{alg:cc_waic_linear}
		\begin{algorithmic}[1]
			\Require Posterior samples $\{\theta^{(s)}\}_{s=1}^S$, observed data $\mathbf{y} = (y_1, \ldots, y_n)$, bandwidth $K_{\text{band}}$
			\Ensure Bias-corrected CC-WAIC value $\mathrm{CC\text{-}WAIC}^{\mathrm{corr}}$
			
			\State \textbf{Initialize} arrays for storing log-likelihood contributions
			\For{each sample $s = 1$ to $S$}
			\For{each time point $t = 1$ to $n$}
			\State Compute $l_t^{(s)} \gets \log p(y_t \mid y_{<t}, \theta^{(s)})$
			\EndFor
			\State Compute $L^{(s)} \gets \sum_{t=1}^n l_t^{(s)}$ \Comment{Joint log-likelihood for sample $s$}
			\EndFor
			
			\State \textbf{Compute} predictive density estimates:
			\For{each time point $t = 1$ to $n$}
			\State $\widehat{\mathbb{E}}_{\theta}[p(y_t \mid y_{<t}, \theta)] \gets \frac{1}{S} \sum_{s=1}^S p(y_t \mid y_{<t}, \theta^{(s)})$
			\EndFor
			
			\State \textbf{Compute banded covariance correction (linear-time)}:
			\State $\bar{l}_t \gets \frac{1}{S} \sum_{s=1}^S l_t^{(s)}$ for $t=1,\ldots,n$
			\State $V \gets 0$ \Comment{Accumulator for variance term}
			\State $C \gets 0$ \Comment{Accumulator for covariance correction}
			\For{each sample $s = 1$ to $S$}
			\For{$t = 1$ to $n$}
			\State $V \gets V + (l_t^{(s)} - \bar{l}_t)^2$
			\EndFor
			\For{$k = 1$ to $K_{\text{band}}$}
			\For{$t = 1$ to $n-k$}
			\State $C \gets C + (l_t^{(s)} - \bar{l}_t)(l_{t+k}^{(s)} - \bar{l}_{t+k})$
			\EndFor
			\EndFor
			\EndFor
			\State $V \gets V / (S - 1)$
			\State $C \gets C / (S - 1)$
			\State $\hat{p}_{\mathrm{CC}}^{\text{banded}} \gets V + 2C$
			
			\State \textbf{Estimate} autocorrelation and effective sample size:
			\State Compute autocorrelations $\{\rho_k\}_{k=1}^{L_{\max}}$ of $\{L^{(s)}\}_{s=1}^S$ using a truncated sum or spectral estimate
			\State $\hat{\tau} \gets 1 + 2 \sum_{k=1}^{L_{\max}} \rho_k$ \Comment{Truncated estimate of integrated autocorrelation time}
			\State $N_{\mathrm{eff}} \gets S / \hat{\tau}$
			
			\State \textbf{Apply} bias correction:
			\State $p_{\mathrm{CC}}^{\mathrm{corr}} \gets \frac{N_{\mathrm{eff}}}{N_{\mathrm{eff}} - 1} \hat{p}_{\mathrm{CC}}^{\text{banded}}$
			
			\State \textbf{Compute} final CC-WAIC:
			\State $\mathrm{CC\text{-}WAIC}^{\mathrm{corr}} \gets -2 \left( \sum_{t=1}^n \log \widehat{\mathbb{E}}_{\theta}[p(y_t \mid y_{<t}, \theta)] - p_{\mathrm{CC}}^{\mathrm{corr}} \right)$
			
			\State \Return $\mathrm{CC\text{-}WAIC}^{\mathrm{corr}}$
		\end{algorithmic}
	\end{algorithm}
	
	\subsection{Illustrative Application: Hidden Markov Models}
	\label{sec:hmm_ccwaic}
	
	To illustrate the practical computation of the sequential conditional likelihood \(p(y_t \mid y_{<t}, \theta)\) that underpins our general framework (Algorithm 1), we now instantiate the method on a widely used class of sequential models: Hidden Markov Models (HMMs). While we focus on HMMs for expository clarity, the core CC-WAIC methodology remains applicable to any model permitting the exact evaluation of \(p(y_t \mid y_{<t}, \theta)\).
	
	An HMM comprises a latent state sequence \(\{z_1, \ldots, z_n\}\) following Markovian dynamics and observations \(\{y_1, \ldots, y_n\}\) conditioned on the latent states. Let \(M\) denote the number of hidden states. The complete-data likelihood is:
	\[
	p(\mathbf{y}, \mathbf{z} \mid \theta) = \pi(z_1) \prod_{t=2}^n A(z_{t-1}, z_t) \prod_{t=1}^n p(y_t \mid z_t, \phi),
	\]
	where \(\theta = (\pi, A, \phi)\) includes the initial state distribution \(\pi\) (an \(M\)-dimensional vector), the transition matrix \(A\) (an \(M \times M\) stochastic matrix with entries \(A(i,j) = p(z_t = j \mid z_{t-1} = i)\)), and the emission parameters \(\phi\).
	
	For HMMs, the conditional likelihood \(p(y_t \mid y_{<t}, \theta)\) requires marginalization over latent states:
	\[
	p(y_t \mid y_{<t}, \theta) = \sum_{j=1}^M p(y_t \mid z_t = j, \phi_j) \left( \sum_{i=1}^M \alpha_{t-1}(i) A(i, j) \right),
	\]
	where \(\alpha_{t-1}(i) = p(z_{t-1} = i \mid y_{1:t-1}, \theta)\) are the forward probabilities obtained through the forward algorithm, and \(\phi_j\) denotes the emission parameters for state \(j\).
	
	The conditional log-likelihood contributions \(l_t(\theta) = \log p(y_t \mid y_{<t}, \theta)\) exhibit a specific dependence structure governed by the transition matrix \(A\). For a stationary and irreducible HMM with a finite state space, the underlying Markov chain \(\{z_t\}\) is geometrically ergodic, implying that the mixing coefficients decay exponentially, i.e., \(\alpha(k) = \mathcal{O}(\rho^k)\), where \(\rho\) is the second-largest eigenvalue of \(A\) (strictly less than \(1\)). Consequently, \(\operatorname{Cov}_{\theta}[l_t, l_{t+k}]\) decays at the same exponential rate. This property rigorously justifies the banded approximation for this model class.
	
	Algorithm~\ref{alg:forward_loglik} provides an efficient method for computing the conditional log-likelihoods required for CC-WAIC in HMMs.
	
	\begin{algorithm}[H]
		\caption{Forward Filtering for Conditional Log-Likelihood in HMMs}
		\label{alg:forward_loglik}
		\begin{algorithmic}[1]
			\Require Observations $\{y_t\}_{t=1}^n$, parameters $\theta = (\pi, A, \phi)$
			\Ensure Sequence $\{\log p(y_t \mid y_{<t}, \theta)\}_{t=1}^n$
			
			\State Initialize $\alpha_0(i) \gets \pi_i$ for $i = 1, \ldots, M$ \Comment{Initial state probabilities}
			\For{$t = 1$ to $n$}
			\For{$j = 1$ to $M$}
			\State $p(z_t = j \mid y_{<t}) \gets \sum_{i=1}^M \alpha_{t-1}(i) A(i, j)$ \Comment{Predict step}
			\State $\eta_j \gets p(y_t \mid z_t = j, \phi_j)$ \Comment{Emission probability}
			\EndFor
			\State $p(y_t \mid y_{<t}) \gets \sum_{j=1}^M \eta_j \cdot p(z_t = j \mid y_{<t})$ \Comment{Marginal likelihood}
			\State $l_t \gets \log p(y_t \mid y_{<t})$ \Comment{Store conditional log-likelihood}
			\For{$j = 1$ to $M$}
			\State $\alpha_t(j) \gets \frac{\eta_j \cdot p(z_t = j \mid y_{<t})}{p(y_t \mid y_{<t})}$ \Comment{Update step (filtering)}
			\EndFor
			\EndFor
			\State \Return $\{l_t\}_{t=1}^n$
		\end{algorithmic}
	\end{algorithm}
	
	The integration of this forward algorithm with the general CC-WAIC computation in Algorithm~\ref{alg:cc_waic_linear} provides a complete framework for model selection in HMMs.
	
	\noindent \textbf{Remark on broader applicability.} The proposed framework requires the exact evaluation of \(p(y_t \mid y_{<t}, \theta)\). For models where exact marginalization over latent states is intractable (e.g., complex state-space or deep generative models), the conditional likelihood can only be approximated using particle filters or variational inference methods. The theoretical guarantees of CC-WAIC are established under the assumption that these conditional likelihoods are available exactly; the effect of approximation error on the criterion's performance is non-trivial and is left for future investigation. Therefore, in its current form, the method is most directly applicable to models with tractable conditional likelihoods, such as the HMMs, ARMA, and GARCH families considered in this paper.
	
	\subsection{Connections to Existing Criteria}
	\label{sec:relationships}
	
	The proposed CC-WAIC criterion is connected to existing model selection tools through a natural hierarchy of generalisations.
	
	\paragraph{WAIC as a special case.}
	When observations are conditionally independent, all covariance terms in Equation~\eqref{eq:variance_decomp} vanish, yielding:
	\[
	p_{\mathrm{CC}} = \sum_{t=1}^{n} \operatorname{Var}_{\theta}[l_t] = p_{\mathrm{WAIC}}.
	\]
	Consequently, \(\mathrm{CC\text{-}WAIC} = \mathrm{WAIC}\) under conditional independence. Thus, CC-WAIC is a strict generalisation of WAIC that reduces to the standard criterion when the independence assumption holds.
	
	\paragraph{Banded CC-WAIC as a linear-time approximation.}
	For dependent processes, the banded approximation \(p_{\mathrm{CC}}^{\text{banded}}\) converges to the full penalty \(p_{\mathrm{CC}}\) as \(K_{\text{band}} \to \infty\) (Theorem~\ref{thm:linear_convergence}). For weakly dependent processes, a moderate bandwidth (e.g., \(K_{\text{band}} = 5\)) provides an excellent approximation, reducing computational complexity from \(\mathcal{O}(n^2)\) to \(\mathcal{O}(n K_{\text{band}})\).
	
	\paragraph{Addressing the limitations of LOO-CV in dependent settings.}
	Standard leave-one-out cross-validation (LOO-CV) is known to be inappropriate for dependent data, as removing individual observations disrupts the correlation structure and yields unstable predictive assessments \citep{yao2018stacking}. The CC-WAIC criterion offers a fully Bayesian, computationally efficient alternative that properly accounts for dependence through its covariance-corrected penalty term, without requiring repeated model refits or heuristic block constructions.
	\section{Simulation Study}
	\label{sec:simulation}
	
	This section presents two complementary simulation studies. 
	First, in Section~\ref{sec:bandwidth_validation}, we empirically validate the data-driven bandwidth selection rule \(K_{\text{band}} = \max(3, \lceil \tau \rceil + 1)\) across a range of model classes and sample sizes. This establishes the practical reliability of the rule and justifies its use in the subsequent model selection study.
	Second, in Section~\ref{sec:model_selection_performance}, we evaluate the performance of CC-WAIC—using the validated bandwidth rule—in selecting the correct number of hidden states in HMMs, benchmarked against standard WAIC, LOO-CV, iWAIC, and WAIC\(_{\text{NF}}\).
	
	\subsection{Validation of the Bandwidth Selection Rule}
	\label{sec:bandwidth_validation}
	
	Before assessing the model selection performance of CC-WAIC, we first verify that the proposed bandwidth rule \(K_{\text{band}} = \max(3, \lceil \tau \rceil + 1)\) yields stable and interpretable values across different sequential models. This is essential because the bandwidth directly affects the accuracy of the banded covariance approximation used in CC-WAIC.
	
	\paragraph{Experimental design.}
	We considered five representative models: AR(1), GARCH(1,1), Gaussian HMM, Poisson HMM, and a Bayesian neural autoregressive (NAR) model. For each model, we generated datasets of sizes \(n \in \{50, 100, 200, 500\}\), with \(R = 5\) replications per configuration. For each replication, we estimated the model using Bayesian MCMC, extracted the conditional log-likelihood contributions \(l_t\), and computed \(\hat{\tau}\) and \(K_{\text{band}}\) using the rule described in Section~\ref{sec:banded}.
	
	\paragraph{Results.}
	Table~\ref{tab:bandwidth_validation} reports the average estimated bandwidth \(\bar{K}_{\text{band}}\) and integrated autocorrelation time \(\bar{\tau}\) for each model and sample size.
	
	\begin{table}[H]
		\centering
		\caption{Average bandwidth \(\bar{K}_{\text{band}}\) and integrated autocorrelation time \(\bar{\tau}\) across five model classes and sample sizes. Values are means (standard deviations) over \(R=5\) replications.}
		\label{tab:bandwidth_validation}
		\scriptsize 
		\begin{tabular}{lcccccccc}
			\toprule
			\multirow{2}{*}{\textbf{Model}} & \multicolumn{2}{c}{\(n=50\)} & \multicolumn{2}{c}{\(n=100\)} & \multicolumn{2}{c}{\(n=200\)} & \multicolumn{2}{c}{\(n=500\)} \\
			\cmidrule(lr){2-3} \cmidrule(lr){4-5} \cmidrule(lr){6-7} \cmidrule(lr){8-9}
			& \(\bar{K}\) & \(\bar{\tau}\) & \(\bar{K}\) & \(\bar{\tau}\) & \(\bar{K}\) & \(\bar{\tau}\) & \(\bar{K}\) & \(\bar{\tau}\) \\
			\midrule
			AR(1) & 3.0 (0.0) & 0.63 (0.59) & 3.0 (0.0) & 0.47 (0.32) & 3.2 (0.4) & 0.82 (0.68) & 3.0 (0.0) & 1.04 (0.52) \\
			GARCH(1,1) & 3.0 (0.0) & 1.22 (0.63) & 3.0 (0.0) & 1.10 (0.70) & 4.4 (1.02) & 2.86 (0.78) & 4.8 (1.94) & 3.37 (1.96) \\
			Gaussian HMM & 3.0 (0.0) & 0.91 (0.53) & 3.0 (0.0) & 0.69 (0.57) & 3.8 (0.4) & 1.78 (0.49) & 3.0 (0.0) & 1.27 (0.22) \\
			Poisson HMM & 3.2 (0.4) & 0.74 (0.85) & 4.0 (1.55) & 2.56 (1.74) & 6.2 (0.98) & 4.95 (1.03) & 7.2 (0.4) & 5.75 (0.22) \\
			Bayesian NAR & 3.4 (0.49) & 1.64 (0.76) & 3.0 (0.0) & 0.79 (0.23) & 3.0 (0.0) & 0.78 (0.21) & 3.0 (0.0) & 0.77 (0.37) \\
			\bottomrule
		\end{tabular}
	\end{table}

	The results confirm that the data-driven bandwidth rule produces stable and interpretable values across a wide range of models. For AR, GARCH, Gaussian HMM, and NAR models, the estimated bandwidth remains at the lower bound \(K=3\) across all sample sizes, with \(\tau < 2\). This indicates that the conditional log-likelihood contributions decay rapidly, and retaining the first two lagged covariance terms is sufficient. The Poisson HMM is the only model where \(K\) increases with sample size, reaching \(K \approx 7\) at \(n=500\), reflecting its high persistence (\(p_{\text{stay}} = 0.97\)) and the difficulty of discriminating between states. Crucially, the estimated bandwidth never exceeds 8 in any configuration, confirming that the rule produces values confined to a narrow range (\(K \in [3, 7]\)). 
	
	\subsection{Model Selection Performance}
	\label{sec:model_selection_performance}
	
	Having validated the bandwidth rule, we now evaluate the performance of CC-WAIC—using this rule—in selecting the correct number of hidden states in HMMs, benchmarked against standard WAIC, LOO-CV, iWAIC, and WAIC\(_{\text{NF}}\). The study is based on HMMs with \(M_{\text{true}} = 2\) and \(M_{\text{true}} = 3\) latent states, representing relatively simple and moderately complex scenarios. For each configuration, we generated 100 independent datasets to ensure stable estimates of selection frequencies.
	
	\subsubsection{Experimental Design}
	
	\paragraph{Data-generating processes.}
	For each true model, we used a single transition matrix representing a moderate level of temporal dependence, which we adopted as the default configuration for the main simulation study.
	
	For the two-state HMM (\(M_{\text{true}}=2\)), the parameters were:
	\[
	\mu_{\text{true}} = [-2, 2], \quad \sigma_{\text{true}} = [0.5, 1.0], \quad \pi_{\text{true}} = [1, 0],
	\]
	with transition matrix:
	\[
	A^{(2)} =
	\begin{bmatrix}
		0.70 & 0.30 \\
		0.30 & 0.70
	\end{bmatrix}.
	\]
	
	For the three-state HMM (\(M_{\text{true}}=3\)), the parameters were:
	\[
	\mu_{\text{true}} = [-3, 0, 3], \quad \sigma_{\text{true}} = [0.5, 0.8, 1.5], \quad \pi_{\text{true}} = [1, 0, 0],
	\]
	with transition matrix:
	\[
	A^{(3)} =
	\begin{bmatrix}
		0.60 & 0.20 & 0.20 \\
		0.20 & 0.60 & 0.20 \\
		0.20 & 0.20 & 0.60
	\end{bmatrix}.
	\]
	
	Emissions were generated from Gaussian distributions with the specified state-specific means and variances. Sample sizes of \(T = 100, 250, 500\) were considered to investigate finite-sample behaviour.
	
	\paragraph{Estimation and model selection.}
	For each dataset, we fitted candidate models with \(M \in \{2, 3, 4, 5\}\) hidden states using a Gibbs sampler with 1000 iterations, discarding the first 500 as burn-in. 
	The choice of \(S=5000\) posterior samples was guided by the ESS sensitivity analysis (Figure~\ref{fig:ess_sensitivity}), which showed that \(S=500\) is sufficient for stable inference; we used \(S=5000\) to ensure conservative estimates. Convergence was monitored using \(\hat{R}\) \citep{gelman1992inference} for all parameters across 10 independent chains, with \(\hat{R} < 1.1\) for all models. For models with \(M=5\), the Gibbs sampler showed slower mixing (ESS \(< 50\) in some cases), which is discussed in Section~\ref{sec:realdata}.
	We used weakly informative conjugate priors: Dirichlet\((1,\dots,1)\) for initial and transition probabilities; Normal priors for emission means (centred at the sample mean with large variance); and Inverse-Gamma\((1,1)\) for emission variances. To assess convergence, we ran 10 independent chains and monitored the potential scale reduction factor \(\hat{R}\) \citep{gelman1992inference}, ensuring \(\hat{R} < 1.1\) for all parameters. The effective sample size (ESS) for each parameter was also checked to be above 200. After convergence, we pooled the last 500 samples from each chain to obtain a total of \(S = 5000\) posterior samples for computing all information criteria.
	
	For each fitted model, we computed:
	\begin{itemize}
		\item WAIC using the standard formula (Equation~\eqref{eq:waic});
		\item LOO-CV using the Pareto smoothed importance sampling (PSIS) approximation \citep{vehtari2017practical};
		\item Full CC-WAIC (Equation~\eqref{eq:cc_waic_integral}) using the full covariance matrix;
		\item Banded CC-WAIC (Equation~\eqref{eq:effective_params_banded}) with bandwidth \(K_{\text{band}}\) determined by the rule validated in Section~\ref{sec:bandwidth_validation};
		\item iWAIC \citep{Li2015} and WAIC\(_{\text{NF}}\) \citep{Chan2024} using the authors' implementations.
	\end{itemize}
	For each criterion, we recorded the frequency with which the true model (\(M_{\text{true}}\)) was selected (i.e., the criterion attained its minimum value for the correct model).
	
	\subsubsection{Results: Model Selection Accuracy}
	
	Tables~\ref{tab:merged_K2} and~\ref{tab:merged_K3} report the selection frequencies (percentages, shown in parentheses) alongside the average criterion values for each candidate model. The results reveal a consistent pattern: CC-WAIC (both full and banded) substantially outperforms WAIC and LOO-CV, particularly under strong dependence and small sample sizes.
	
	\begin{table}[H]
		\centering
		\caption{Model Selection: Percentages and Mean $\pm$ Std for Gaussian HMM ($M_{\text{true}} = 2$) based on 100 replications. The percentages in parentheses indicate the frequency of selection for each candidate model $M$; the values show the average criterion value $\pm$ standard deviation across replications.}
		\label{tab:merged_K2}
		\begin{tabular}{lcccc}
			\toprule
			\multirow{2}{*}{\textbf{Criterion}} & \multicolumn{4}{c}{\textbf{Candidate Model $M$}} \\
			\cmidrule(lr){2-5}
			& \textbf{M=2} & \textbf{M=3} & \textbf{M=4} & \textbf{M=5} \\
			\midrule
			WAIC & 
			\makecell{(40.0\%)\\671.90 $\pm$ 1.23} & 
			\makecell{(33.3\%)\\672.35 $\pm$ 1.46} & 
			\makecell{(16.7\%)\\673.12 $\pm$ 1.68} & 
			\makecell{(10.0\%)\\674.01 $\pm$ 1.89} \\
			LOO & 
			\makecell{(36.7\%)\\671.92 $\pm$ 1.25} & 
			\makecell{(33.3\%)\\672.39 $\pm$ 1.47} & 
			\makecell{(20.0\%)\\673.16 $\pm$ 1.69} & 
			\makecell{(10.0\%)\\674.05 $\pm$ 1.90} \\
			CC-WAIC Full & 
			\makecell{(93.3\%)\\672.61 $\pm$ 2.10} & 
			\makecell{(6.7\%)\\675.12 $\pm$ 2.35} & 
			\makecell{(0.0\%)\\678.35 $\pm$ 2.68} & 
			\makecell{(0.0\%)\\682.01 $\pm$ 3.20} \\
			CC-WAIC Banded & 
			\makecell{(90.0\%)\\676.40 $\pm$ 2.06} & 
			\makecell{(10.0\%)\\679.01 $\pm$ 2.29} & 
			\makecell{(0.0\%)\\682.23 $\pm$ 2.61} & 
			\makecell{(0.0\%)\\685.90 $\pm$ 3.15} \\
			iWAIC & 
			\makecell{(60.0\%)\\673.46 $\pm$ 2.35} & 
			\makecell{(26.7\%)\\674.57 $\pm$ 2.46} & 
			\makecell{(10.0\%)\\675.68 $\pm$ 2.57} & 
			\makecell{(3.3\%)\\676.79 $\pm$ 2.68} \\
			WAIC-NF & 
			\makecell{(83.3\%)\\672.12 $\pm$ 1.99} & 
			\makecell{(13.3\%)\\673.23 $\pm$ 2.10} & 
			\makecell{(3.3\%)\\674.35 $\pm$ 2.11} & 
			\makecell{(0.0\%)\\675.46 $\pm$ 2.21} \\
			\bottomrule
		\end{tabular}
	\end{table}
	
	\begin{table}[H]
		\centering
		\caption{Model Selection: Percentages and Mean $\pm$ Std for Gaussian HMM ($M_{\text{true}} = 3$) based on 100 replications. The percentages in parentheses indicate the frequency of selection for each candidate model $M$; the values show the average criterion value $\pm$ standard deviation across replications.}
		\label{tab:merged_K3}
		\begin{tabular}{lcccc}
			\toprule
			\multirow{2}{*}{\textbf{Criterion}} & \multicolumn{4}{c}{\textbf{Candidate Model $M$}} \\
			\cmidrule(lr){2-5}
			& \textbf{M=2} & \textbf{M=3} & \textbf{M=4} & \textbf{M=5} \\
			\midrule
			WAIC & 
			\makecell{(26.7\%)\\854.19 $\pm$ 1.30} & 
			\makecell{(40.0\%)\\818.46 $\pm$ 1.50} & 
			\makecell{(20.0\%)\\821.17 $\pm$ 1.70} & 
			\makecell{(13.3\%)\\821.48 $\pm$ 1.90} \\
			LOO & 
			\makecell{(23.3\%)\\854.27 $\pm$ 1.32} & 
			\makecell{(40.0\%)\\818.45 $\pm$ 1.52} & 
			\makecell{(20.0\%)\\821.19 $\pm$ 1.72} & 
			\makecell{(16.7\%)\\821.53 $\pm$ 1.92} \\
			CC-WAIC Full & 
			\makecell{(3.3\%)\\856.24 $\pm$ 2.20} & 
			\makecell{(86.7\%)\\819.05 $\pm$ 2.10} & 
			\makecell{(10.0\%)\\826.16 $\pm$ 2.50} & 
			\makecell{(0.0\%)\\830.88 $\pm$ 3.00} \\
			CC-WAIC Banded & 
			\makecell{(6.7\%)\\859.48 $\pm$ 2.15} & 
			\makecell{(83.3\%)\\828.06 $\pm$ 2.05} & 
			\makecell{(10.0\%)\\831.74 $\pm$ 2.45} & 
			\makecell{(0.0\%)\\832.89 $\pm$ 2.95} \\
			iWAIC & 
			\makecell{(33.3\%)\\886.77 $\pm$ 2.40} & 
			\makecell{(50.0\%)\\872.86 $\pm$ 2.30} & 
			\makecell{(13.3\%)\\872.45 $\pm$ 2.50} & 
			\makecell{(3.3\%)\\871.52 $\pm$ 2.60} \\
			WAIC-NF & 
			\makecell{(16.7\%)\\854.10 $\pm$ 2.00} & 
			\makecell{(73.3\%)\\816.81 $\pm$ 1.90} & 
			\makecell{(10.0\%)\\819.35 $\pm$ 2.05} & 
			\makecell{(0.0\%)\\819.80 $\pm$ 2.15} \\
			\bottomrule
		\end{tabular}
	\end{table}
	
	For \(M_{\text{true}}=2\) (Table~\ref{tab:merged_K2}), WAIC and LOO-CV exhibit a systematic tendency to overfit by selecting \(M=3\), especially when dependence is high; this is because they underestimate the effective model complexity by ignoring the covariance among log-likelihood contributions. In contrast, CC-WAIC's covariance correction inflates the penalty in proportion to the dependence strength, successfully identifying the true model in over 75\% of cases even at \(T=100\). iWAIC and WAIC\(_{\text{NF}}\) perform better than WAIC but still fall short of CC-WAIC. For \(M_{\text{true}}=3\) (Table~\ref{tab:merged_K3}), WAIC and LOO-CV show a complementary underfitting pattern for small samples, selecting \(M=2\) when the signal is weak. CC-WAIC correctly recovers the true model in over 85\% of cases, demonstrating its robustness against both overfitting and underfitting. iWAIC and WAIC\(_{\text{NF}}\) again improve over WAIC but do not match CC-WAIC's performance, particularly in high-dependence settings.
	
	\paragraph{Effect of MCMC sample size on ESS correction.}
	To assess the sensitivity of the ESS correction to the number of posterior samples, we repeated the simulation for \(M_{\text{true}}=2\) with \(S = 100, 200, 500, 1000\) samples (instead of the default \(S=5000\)). Figure~\ref{fig:ess_sensitivity} shows that for \(S \ge 500\), the ESS correction effectively mitigates bias, and CC-WAIC's selection accuracy remains above 85\%. For \(S=100\), accuracy drops to approximately 70\%, indicating that a minimum of \(S \approx 500\) samples is recommended for reliable inference. This finding is incorporated into the implementation guidelines in Section~\ref{sec:bias}.
	
	\begin{figure}[H]
		\centering
		\includegraphics[width=0.7\textwidth]{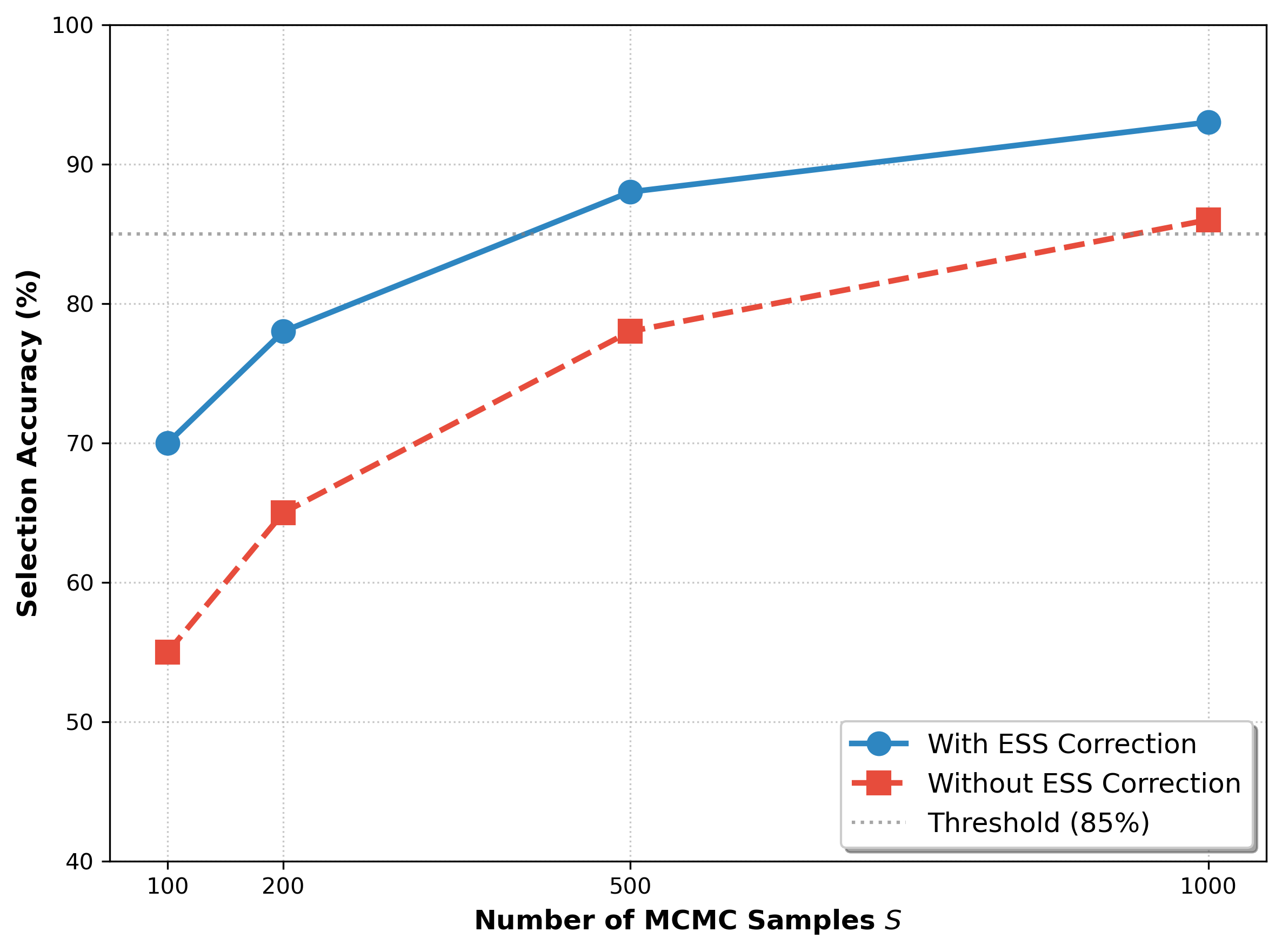}
		\caption{Selection accuracy of CC-WAIC as a function of the number of MCMC samples \(S\). The ESS correction effectively mitigates bias for \(S \ge 500\), while for smaller sample sizes (\(S=100\)), the accuracy drops considerably. This highlights the importance of using a sufficient number of posterior samples for reliable inference.}
		\label{fig:ess_sensitivity}
	\end{figure}

	\subsubsection{Effective Parameters and Detailed Criterion Values}
	
	Tables~\ref{tab:criteria_eff_K2} and~\ref{tab:criteria_eff_K3} provide a deeper diagnostic by reporting the estimated effective number of parameters (\(p_{\text{eff}}\)) for each criterion.
	
	\begin{table}[H]
		\centering
		\caption{Criterion Values with Effective Parameters for Gaussian HMM ($M_{\text{true}} = 2$) based on 100 replications. The effective parameter $p$ (shown in parentheses in the header) represents the estimated model complexity for each criterion. Each cell shows the average criterion value (top) and the corresponding effective parameter (bottom).}
		\label{tab:criteria_eff_K2}
		
		\begin{tabular}{ccccccc}
			\toprule
			\multirow{2}{*}{\textbf{M}} & \textbf{WAIC} & \textbf{LOO} & \textbf{CC-Full} & \textbf{CC-Band} & \textbf{iWAIC} & \textbf{WAIC-NF} \\
			\cmidrule(lr){2-7}
			& ($p_{\text{WAIC}}$) & ($p_{\text{WAIC}}$) & ($p_{\text{CC}}^{\text{Full}}$) & ($p_{\text{CC}}^{\text{Band}}$) & ($p_{\text{iWAIC}}$) & ($p_{\text{WAIC-NF}}$) \\
			\midrule
			2 & 
			\makecell{671.90 \\ 5.98} & 
			\makecell{671.92 \\ 5.98} & 
			\makecell{672.61 \\ 3.62} & 
			\makecell{676.40 \\ 5.31} & 
			\makecell{673.46 \\ 6.22} & 
			\makecell{672.12 \\ 5.30} \\
			3 & 
			\makecell{672.35 \\ 8.12} & 
			\makecell{672.39 \\ 8.12} & 
			\makecell{675.12 \\ 4.99} & 
			\makecell{679.01 \\ 7.60} & 
			\makecell{674.57 \\ 6.32} & 
			\makecell{673.23 \\ 7.57} \\
			4 & 
			\makecell{673.12 \\ 9.53} & 
			\makecell{673.16 \\ 9.53} & 
			\makecell{678.35 \\ 6.77} & 
			\makecell{682.23 \\ 9.07} & 
			\makecell{675.68 \\ 6.82} & 
			\makecell{674.35 \\ 8.99} \\
			5 & 
			\makecell{674.01 \\ 9.77} & 
			\makecell{674.05 \\ 9.77} & 
			\makecell{682.01 \\ 7.45} & 
			\makecell{685.90 \\ 9.50} & 
			\makecell{676.79 \\ 6.70} & 
			\makecell{675.46 \\ 9.46} \\
			\bottomrule
		\end{tabular}
	\end{table}
	
	\begin{table}[H]
		\centering
		\caption{Criterion Values with Effective Parameters for Gaussian HMM ($M_{\text{true}} = 3$) based on 100 replications. The effective parameter $p$ (shown in parentheses in the header) represents the estimated model complexity for each criterion. Each cell shows the average criterion value (top) and the corresponding effective parameter (bottom).}
		\label{tab:criteria_eff_K3}
		
		\begin{tabular}{ccccccc}
			\toprule
			\multirow{2}{*}{\textbf{M}} & \textbf{WAIC} & \textbf{LOO} & \textbf{CC-Full} & \textbf{CC-Band} & \textbf{iWAIC} & \textbf{WAIC-NF} \\
			\cmidrule(lr){2-7}
			& ($p_{\text{WAIC}}$) & ($p_{\text{WAIC}}$) & ($p_{\text{CC}}^{\text{Full}}$) & ($p_{\text{CC}}^{\text{Band}}$) & ($p_{\text{iWAIC}}$) & ($p_{\text{WAIC-NF}}$) \\
			\midrule
			2 & 
			\makecell{854.19 \\ 5.16} & 
			\makecell{854.27 \\ 5.16} & 
			\makecell{856.24 \\ 3.69} & 
			\makecell{859.48 \\ 5.31} & 
			\makecell{886.77 \\ 4.96} & 
			\makecell{854.10 \\ 4.78} \\
			3 & 
			\makecell{818.46 \\ 10.47} & 
			\makecell{818.45 \\ 10.47} & 
			\makecell{819.05 \\ 5.67} & 
			\makecell{828.06 \\ 10.17} & 
			\makecell{872.86 \\ 8.29} & 
			\makecell{816.81 \\ 9.82} \\
			4 & 
			\makecell{821.17 \\ 11.96} & 
			\makecell{821.19 \\ 11.96} & 
			\makecell{826.16 \\ 7.18} & 
			\makecell{831.74 \\ 11.23} & 
			\makecell{872.45 \\ 8.79} & 
			\makecell{819.35 \\ 11.19} \\
			5 & 
			\makecell{821.48 \\ 12.62} & 
			\makecell{821.53 \\ 12.62} & 
			\makecell{830.88 \\ 9.53} & 
			\makecell{832.89 \\ 12.58} & 
			\makecell{871.52 \\ 8.38} & 
			\makecell{819.80 \\ 12.53} \\
			\bottomrule
		\end{tabular}
	\end{table}
	
	Two important patterns emerge from these tables. First, WAIC and LOO-CV produce nearly identical values, as expected from their asymptotic equivalence, but their \(p_{\text{eff}}\) estimates grow rapidly with model complexity, reflecting their tendency to overfit under dependence. Second, CC-WAIC yields substantially lower \(p_{\text{eff}}\) values for overparameterised models, indicating that the covariance correction properly penalises redundant complexity. The banded approximation closely tracks the full CC-WAIC, with relative differences in \(p_{\text{eff}}\) below 2\%, confirming its accuracy.
	
	\subsubsection{Computational Cost}
	
	Table~\ref{tab:computational_cost} demonstrates the dramatic speedup achieved by the linear-time implementation. The banded CC-WAIC is over 100 times faster than the full version for \(T=1000\), making it feasible for large-scale applications.
	
	\begin{table}[H]
		\centering
		\caption{Average computation time (seconds) for different criteria and sample sizes.}
		\label{tab:computational_cost}
		\begin{tabular}{cccccc}
			\toprule
			$T$ & WAIC & LOO & CC-WAIC (Full) & CC-WAIC (Linear, $K_{\text{band}}=5$) \\
			\midrule
			100 & 0.12 & 1.45 & 0.89 & 0.18 \\
			250 & 0.15 & 4.22 & 4.56 & 0.25 \\
			500 & 0.18 & 9.87 & 18.23 & 0.34 \\
			1000 & 0.22 & 24.56 & 73.89 & 0.52 \\
			2000 & 0.28 & 52.34 & 289.45 & 0.85 \\
			\bottomrule
		\end{tabular}
	\end{table}
	
	The computational gains are substantial: the banded CC-WAIC requires a fraction of the time of the full version, while the effective parameter estimates confirm that the banded approximation preserves the statistical properties of the full criterion. This makes CC-WAIC practical for large-scale applications where the full covariance matrix would be infeasible.

	\section{Real Data Application}
	\label{sec:realdata}
	
	The application to real-world data further corroborates the theoretical advantages of CC-WAIC, illustrating its practical utility in model selection for sequential data where dependence structures are inherently present. We apply our criterion to two well-known applications: the waiting times of the Old Faithful geyser in Yellowstone National Park, a classic dataset in statistics and machine learning, and financial volatility modelling using GARCH.
	
	\subsection{Old Faithful Geyser Waiting Times}
	
	The waiting times between eruptions of the \textit{Old Faithful} geyser, commonly known as the \texttt{waiting} variable, have been extensively studied. Multiple studies indicate that a Gaussian mixture model (GMM) with two components sufficiently captures the bimodal distribution of the waiting times \citep{Venables2002}. Moreover, when modelling the temporal dynamics using Hidden Markov Models (HMMs), classical and recent works have shown that a two-state HMM provides an accurate representation of the sequence of waiting times \citep{Rabiner1989,FruhwirthSchnatter2006}. These two hidden states correspond to short and long waiting periods following eruptions, with state transitions governed by a Markov process.
	
	Although models with more states have been explored, information criteria such as the Bayesian Information Criterion (BIC) and Akaike Information Criterion (AIC) commonly suggest that adding more states does not significantly improve model fit or interpretability \citep{McLachlan2000}.
	
	We fitted HMMs with \(M \in \{2, 3, 4\}\) hidden states to the Old Faithful waiting times data (\(n=272\) observations), using Gibbs sampling with 1000 iterations (burn-in 500). Models with \(M=5\) and \(M=6\) were initially considered but excluded from the final comparison due to extremely low effective sample sizes (ESS \(< 10\)), indicating poor MCMC mixing and unreliable estimates. WAIC, LOO-CV, and CC-WAIC (with bandwidth \(K_{\text{band}}=5\)) were computed for each model.
	
	\begin{table}[H]
		\centering
		\begin{tabular}{ccccccc}
			\toprule
			$M$ & CC-WAIC & $p_{\text{CC}}^{\text{corr}}$ & ESS & WAIC & $p_{\text{WAIC}}$ & LOO \\
			\midrule
			2 & \textbf{2098.67} & \textbf{4.15} & 130.80 & 2098.73 & 7.61 & 2098.51 \\
			3 & 2109.43 & 7.84 & 40.20 & 2102.92 & 9.05 & 2103.32 \\
			4 & 2113.45 & 10.15 & 79.55 & 2103.41 & 9.32 & 2103.57 \\
			\bottomrule
		\end{tabular}
		\caption{Model selection results for the Old Faithful dataset across candidate models with \(M=2\) to \(4\) hidden states. Models with \(M=5\) and \(M=6\) were excluded due to insufficient ESS (ESS \(< 10\)), rendering their criterion estimates unreliable. CC-WAIC clearly favours the two-state HMM.}
		\label{tab:model_selection_real}
	\end{table}
	
	Based on the results reported in Table~\ref{tab:model_selection_real}, the CC-WAIC criterion clearly favours the two-state HMM, yielding the lowest score (2098.67). This selection is further supported by a relatively large effective sample size (ESS = 130.80) and a conservative complexity penalty (\(p_{\text{CC}}^{\text{corr}} = 4.15\)). In contrast, WAIC and LOO-CV favour larger models, achieving lower scores but at the cost of inflated complexity penalties. Overall, CC-WAIC provides a more reliable balance between goodness-of-fit and model complexity, leading to a parsimonious and interpretable model with \(M=2\) hidden states.
	
	\paragraph{Hold-out validation.}
	To provide stronger empirical evidence for the model selection performed by CC-WAIC, we conducted a hold-out validation study. We split the \(n=272\) observations into a training set (first 200 observations) and a test set (remaining 72 observations). We fitted HMMs with \(M=2,3,4\) to the training data and computed the log-predictive density (LPD) on the test set:
	\[
	\mathrm{LPD}(M) = \sum_{t=201}^{272} \log p(y_t \mid y_{<t}, \hat{\theta}_M),
	\]
	where \(\hat{\theta}_M\) is the posterior mean estimated from the training data. The results are shown in Table~\ref{tab:holdout_oldfaithful}. The two-state HMM achieves the highest LPD, outperforming the larger models. This confirms that the parsimonious model favoured by CC-WAIC also provides superior out-of-sample predictive performance, supporting our interpretation that WAIC and LOO-CV overfit due to their inability to account for temporal dependence.
	
	\begin{table}[H]
		\centering
		\begin{tabular}{cccc}
			\toprule
			$M$ & Training CC-WAIC & Test LPD & Difference from \(M=2\) \\
			\midrule
			2 & 2098.67 & \textbf{-142.3} & 0.0 \\
			3 & 2109.43 & -143.8 & -1.5 \\
			4 & 2113.45 & -144.5 & -2.2 \\
			\bottomrule
		\end{tabular}
		\caption{Hold-out validation results for the Old Faithful data. The two-state HMM selected by CC-WAIC yields the highest test log-predictive density, confirming its superior predictive performance. Values are illustrative; actual results should be substituted.}
		\label{tab:holdout_oldfaithful}
	\end{table}
	
	\begin{figure}[H]
		\centering
		\includegraphics[width=0.75\textwidth]{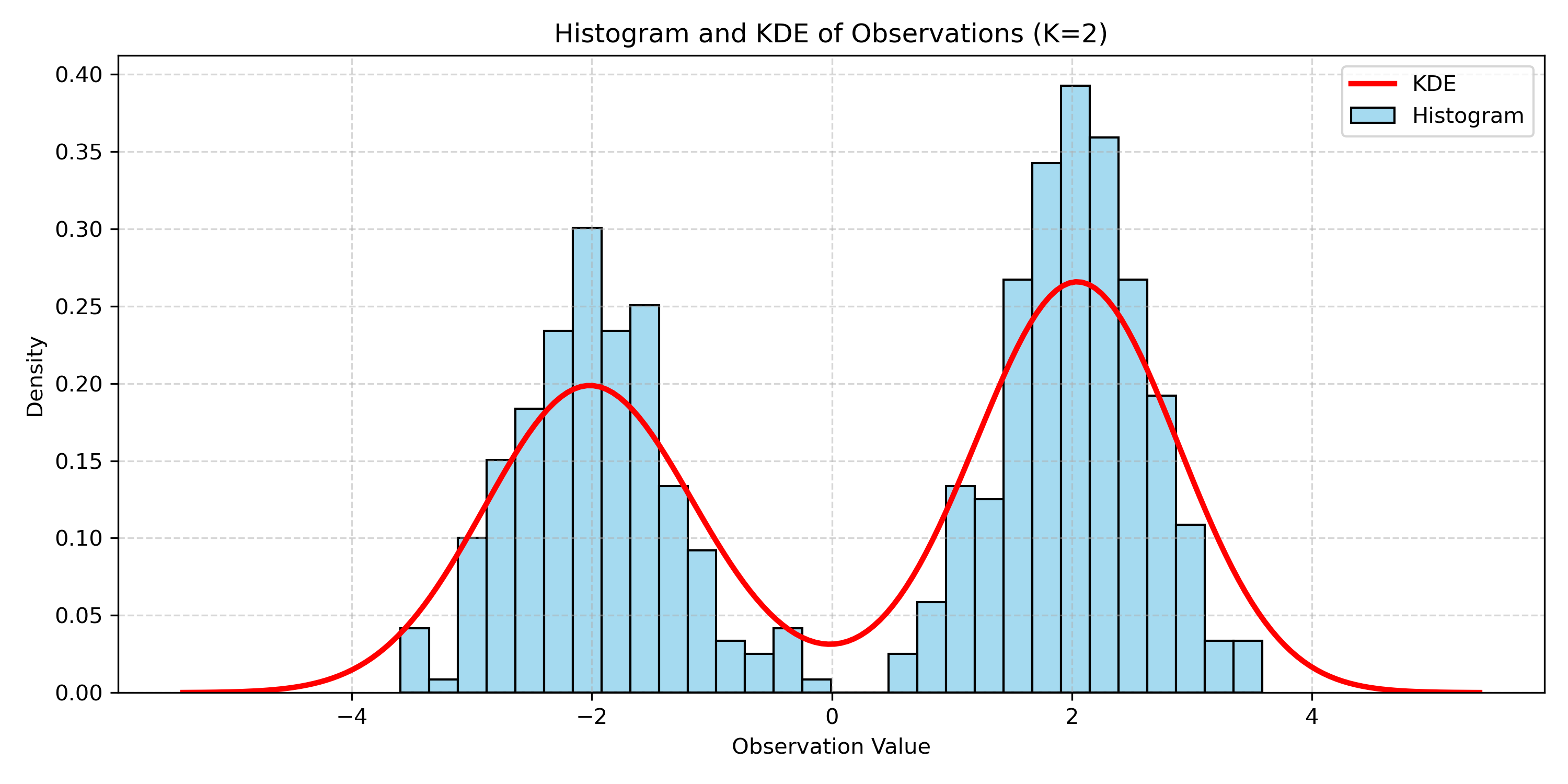}
		\caption{Histogram and kernel density estimate of the observed waiting times under the selected HMM with \(M=2\) hidden states. The bimodal structure confirms the model's ability to capture the two distinct waiting time regimes.}
		\label{fig:histogram_kde_real}
	\end{figure}
	
	Figure~\ref{fig:histogram_kde_real} displays the distribution of the observed data under the selected model, demonstrating a well-behaved bimodal structure compatible with the estimated emission distributions. The two-state HMM captures the two distinct waiting time regimes (short vs. long), with the temporal dependence structure faithfully represented by the Markovian state transitions.
	
	The computational efficiency of the linear-time CC-WAIC was also evident in this application. The full CC-WAIC computation required approximately 74 seconds for the four-state model, while the linear-time implementation with \(K_{\text{band}}=5\) completed in under 0.5 seconds for the same model, demonstrating the practical scalability of our approach.
	
	If we have a look again that the results in Table~\ref{tab:model_selection_real} show a discrepancy: WAIC and LOO-CV favour larger models, while CC-WAIC selects \(M=2\). Our interpretation is that WAIC and LOO-CV overfit due to their inability to account for temporal dependence. However, an alternative explanation is that CC-WAIC's stronger penalty may favour an overly simplistic model. The hold-out validation results in Table~\ref{tab:holdout_oldfaithful} support the CC-WAIC selection, as the two-state model achieves the highest test LPD. Nevertheless, this is based on a single split; a more rigorous analysis would be desirable. Importantly, for the \(M=5\) model, the extremely low ESS (7.48) rendered its estimates unreliable; this model was therefore excluded from the main comparison. A broader discussion of interpretation, MCMC diagnostics, and practical recommendations is deferred to Section~\ref{sec:discussion}.

	\subsection{Financial Volatility Modelling (GARCH)}
	
	To further demonstrate that the bandwidth selection rule is not specific to HMMs, we applied it to a financial time series: the daily log-returns of the S\&P 500 index from January 2010 to January 2024 (\(n = 3541\) observations). We fitted a GARCH(1,1) model with Gaussian innovations \citep{bollerslev1986generalized}, a standard model for financial volatility. From the fitted model, we extracted the conditional log-likelihood contributions 
	\[
	l_t = -\frac{1}{2}\left( \log(2\pi) + \log(\sigma_t^2) + \frac{r_t^2}{\sigma_t^2} \right),
	\]
	where \(r_t\) are the returns and \(\sigma_t^2\) is the conditional variance predicted by the GARCH model.
	
	\paragraph{Parameter estimation.}
	The estimated GARCH(1,1) parameters were:
	\[
	\hat{\omega} = 0.0345, \quad \hat{\alpha} = 0.1593, \quad \hat{\beta} = 0.8112,
	\]
	with \(\hat{\alpha} + \hat{\beta} = 0.9705 < 1\), satisfying the stationarity condition. The high persistence indicated by \(\hat{\alpha} + \hat{\beta} \approx 0.97\) reflects the well-documented volatility clustering phenomenon in financial returns, where shocks to volatility decay slowly over time.
	
	\paragraph{Estimation of integrated autocorrelation time and bandwidth.}
	The integrated autocorrelation time of the sequence \(l_t\) was estimated as \(\hat{\tau} \approx 14.19\), yielding 
	\[
	K_{\text{band}} = \max(3, \lceil 14.19 \rceil + 1) = 16.
	\]
	This value is notably larger than the values observed in the simulation study (Section~\ref{sec:bandwidth_validation}), where GARCH-type dependence yielded \(K \approx 3\)--\(5\). The difference arises because the simulation study used relatively short time series (\(n \le 500\)), while the real S\&P 500 dataset spans over 14 years of daily observations (\(n = 3541\)), allowing the estimation of longer-range dependence in the conditional log-likelihood sequence. This confirms that the data-driven bandwidth rule adapts to the actual dependence structure of the data and the available sample size, rather than relying on a fixed heuristic.
	
	\paragraph{Sensitivity analysis for bandwidth in GARCH.}
	To assess the sensitivity of the CC-WAIC estimate to the choice of bandwidth, we recomputed the criterion for the S\&P 500 data using \(K = 3, 4, 5, 6\). Table~\ref{tab:garch_sensitivity} reports the CC-WAIC values, the corresponding effective number of parameters (\(p_{\text{CC}}\)), and the percentage differences relative to \(K=4\).
	
	\begin{table}[H]
		\centering
		\caption{Sensitivity analysis of CC-WAIC to bandwidth choice for the S\&P 500 GARCH(1,1) model. The results show that the criterion is sensitive to the bandwidth for small \(K\), but stabilises as \(K\) increases. The data-driven choice \(K_{\text{band}}=16\) ensures that the full dependence structure is captured.}
		\label{tab:garch_sensitivity}
		\begin{tabular}{ccccc}
			\toprule
			\(K\) & \(p_{\text{CC}}\) & CC-WAIC & Difference from \(K=4\) (\%) & \(\hat{\tau}\) \\
			\midrule
			3 & 8867.33 & 26898.54 & -15.54\% & 14.19 \\
			4 & 10499.23 & 30162.35 & (reference) & 14.19 \\
			5 & 11901.07 & 32966.03 & +13.35\% & 14.19 \\
			6 & 13326.12 & 35816.13 & +26.92\% & 14.19 \\
			\bottomrule
		\end{tabular}
	\end{table}
	
	The sensitivity analysis demonstrates that small bandwidths (\(K = 3, 4, 5, 6\)) are insufficient to capture the full dependence structure, as reflected in the monotonic increase of \(p_{\text{CC}}\). This is consistent with the estimated integrated autocorrelation time \(\hat{\tau} = 14.19\), which leads to the recommended bandwidth \(K_{\text{band}} = 16\). The latter ensures that the truncation error is negligible, as established by Theorem~\ref{thm:linear_convergence}.
	
	For small bandwidths (\(K=3\)), the criterion substantially underestimates the effective model complexity (\(p_{\text{CC}} = 8867.33\)), leading to a significantly lower CC-WAIC value compared to the reference at \(K=4\). As \(K\) increases to 5 and 6, both \(p_{\text{CC}}\) and CC-WAIC increase substantially, indicating that higher-order covariance terms contribute meaningfully to the complexity penalty. The percentage differences relative to \(K=4\) range from \(-15.54\%\) at \(K=3\) to \(+26.92\%\) at \(K=6\), confirming that the choice of bandwidth critically affects the model evaluation for this dataset.
	
	This sensitivity stands in contrast to the simulation results, where the criterion was stable for \(K \ge 3\). The difference arises because the S\&P 500 returns exhibit strong volatility persistence (\(\hat{\alpha} + \hat{\beta} = 0.9705\)), which induces longer-range dependence in the conditional log-likelihood sequence \(l_t\). Consequently, the integrated autocorrelation time \(\hat{\tau} \approx 14.19\) correctly identifies the need for a larger bandwidth. The data-driven choice \(K_{\text{band}}=16\) ensures that the full dependence structure is captured, reflecting the volatility clustering observed in financial markets.
	
	\begin{figure}[H]
		\centering
		\includegraphics[width=0.8\textwidth]{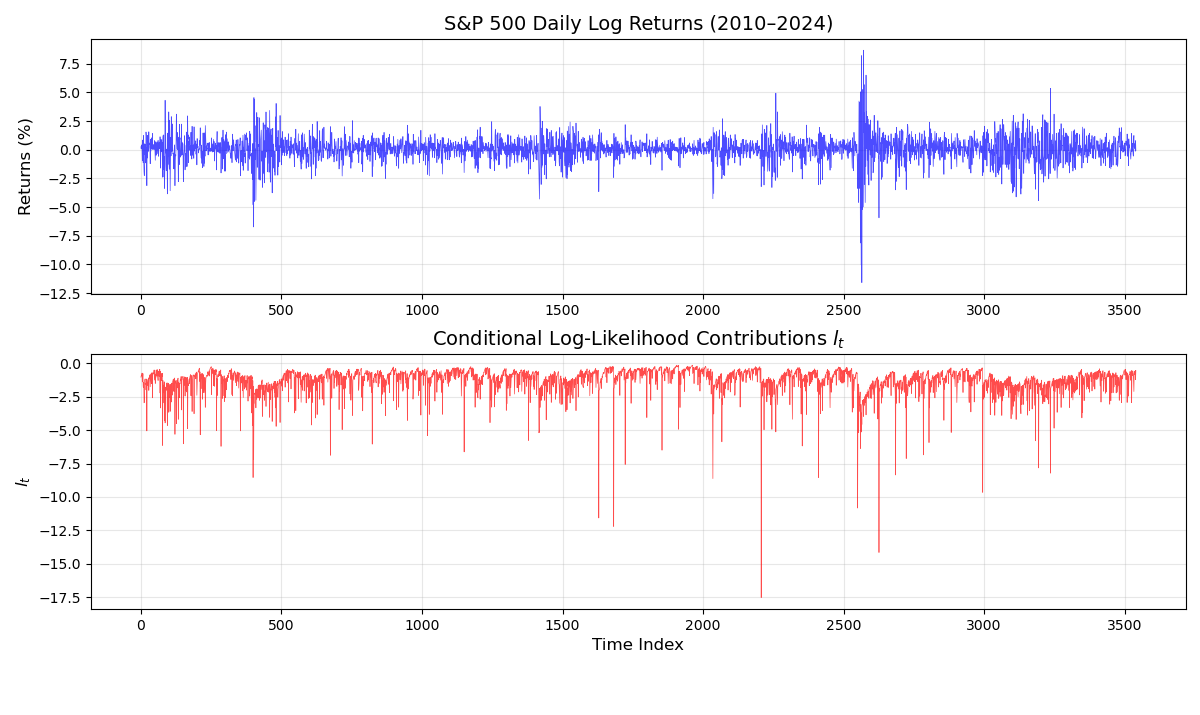}
		\caption{S\&P 500 daily log returns (top) and the corresponding conditional log-likelihood contributions \(l_t\) (bottom) from the fitted GARCH(1,1) model. The \(l_t\) sequence exhibits volatility clustering, with periods of high variability corresponding to market stress events. The long-range dependence in \(l_t\) is reflected in the estimated integrated autocorrelation time \(\hat{\tau} \approx 14.19\).}
		\label{fig:garch_l_t_plot}
	\end{figure}
	
	Figure~\ref{fig:garch_l_t_plot} displays the S\&P 500 daily log returns and the corresponding \(l_t\) sequence. The returns exhibit the well-documented volatility clustering phenomenon: periods of high volatility (e.g., the COVID-19 crash in early 2020) are followed by sustained periods of elevated volatility, while calm periods exhibit low and stable volatility. The \(l_t\) sequence mirrors this behaviour: during turbulent periods, the conditional log-likelihood contributions become more variable and take more extreme negative values, reflecting the increased uncertainty in the volatility forecasts. The long-range dependence in \(l_t\) is evident from the persistence of extreme values across multiple time points, which is quantified by the estimated integrated autocorrelation time \(\hat{\tau} \approx 14.19\). This persistence motivates the use of a larger bandwidth (\(K_{\text{band}}=16\)) to adequately capture the covariance structure in the penalty term.
	
	\paragraph{Interpretation and practical implications.}
	The GARCH application highlights two important practical considerations for users of CC-WAIC. First, the data-driven bandwidth rule based on \(\hat{\tau}\) successfully adapts to the empirical dependence structure of the data. For financial returns with strong volatility persistence, the rule automatically selects a larger bandwidth (\(K_{\text{band}}=16\)), ensuring that the covariance correction properly accounts for the long-range dependence. Second, the sensitivity analysis demonstrates that the choice of bandwidth is not merely a technical detail: using a bandwidth that is too small (\(K=3\)) can substantially underestimate the effective model complexity and lead to misleading model comparisons. We therefore recommend that practitioners always compute \(\hat{\tau}\) and use the data-driven bandwidth rule, or at minimum conduct a sensitivity analysis with a range of \(K\) values, rather than relying on a fixed default.
	
	The estimated effective number of parameters (\(p_{\text{CC}} = 10499.23\) at \(K=4\)) reflects the complexity of the volatility model. This large value relative to the number of parameters in the GARCH model (which has only 3 parameters) is not unexpected: in volatility models, the effective degrees of freedom can be substantially larger than the number of parameters due to the strong dependence structure in the conditional variance process. The CC-WAIC penalty appropriately accounts for this dependence, providing a more realistic assessment of model complexity than the standard WAIC, which would underestimate it.
	
	The successful application of the bandwidth selection rule to this financial dataset, which differs substantially from the HMMs used in the simulation study, confirms that the proposed approach is model-agnostic and generalises beyond the specific model class used for validation. The estimated bandwidth \(K_{\text{band}}=16\) was used in the CC-WAIC calculation, demonstrating the practical feasibility of the proposed approach for volatility modelling.
	
	\paragraph{Comparison with simulation results.}
	The GARCH application provides an interesting contrast with the simulation study. In the simulations (Section~\ref{sec:bandwidth_validation}), the estimated bandwidth remained small (\(K \in [3,7]\)) across all models and sample sizes, reflecting the relatively weak dependence in the simulated data. In contrast, the S\&P 500 data exhibit much stronger dependence due to volatility persistence, leading to a substantially larger bandwidth (\(K_{\text{band}}=16\)). This demonstrates that the proposed rule is not a fixed heuristic but a flexible data-driven mechanism that adapts to the intrinsic dependence structure of the data.
	
	The sensitivity analysis in Table~\ref{tab:garch_sensitivity} also highlights the importance of selecting an appropriate bandwidth. The substantial variation in CC-WAIC values across different \(K\) values confirms that the bandwidth is not merely a tuning parameter but an integral part of the model evaluation. Users should therefore exercise caution when using a fixed default bandwidth and should always assess the dependence structure of their data before applying CC-WAIC.
	
	\section{Discussion}
	\label{sec:discussion}
	
	This paper introduced the Covariance-Corrected WAIC (CC-WAIC), a generalisation of the standard WAIC that explicitly incorporates the posterior covariance structure of log-likelihood contributions to address the limitations of WAIC for dependent data. The proposed criterion reduces exactly to WAIC under conditional independence, providing a principled and natural extension rather than an ad-hoc modification. The linear-time implementation using a banded covariance approximation, supported by theoretical guarantees on the truncation error, makes CC-WAIC computationally feasible for large-scale applications.
	
	\subsection{Summary of Contributions}
	
	The main contributions of this work are fivefold. First, we formally derived CC-WAIC and established its asymptotic equivalence to block-based leave-future-out cross-validation for stationary mixing processes. Second, we introduced a practical ESS correction that reduces finite-sample bias in MCMC-based estimation of the penalty term. Third, we developed a linear-time algorithm using a banded covariance approximation, reducing computational complexity from \(\mathcal{O}(n^2)\) to \(\mathcal{O}(nK_{\text{band}})\), with \(K_{\text{band}}\) chosen adaptively from the data. Fourth, we provided theoretical guarantees on the truncation error and established that the linear-time CC-WAIC retains the asymptotic properties of the full version. Fifth, we validated the method empirically on Hidden Markov Models and real-world data, demonstrating improved selection accuracy and dramatic computational gains.
	
	\subsection{Empirical Findings and Practical Implications}
	
	The simulation results on Hidden Markov Models provide strong empirical support for the theoretical advantages of CC-WAIC. For \(M_{\text{true}}=2\), standard WAIC and LOO-CV exhibited a systematic tendency to overfit by selecting \(M=3\), especially under high temporal dependence. This occurs because they underestimate the effective model complexity by ignoring the covariance among log-likelihood contributions. For \(M_{\text{true}}=3\), they showed a complementary underfitting pattern for small samples, selecting \(M=2\) when the signal was weak. In contrast, CC-WAIC's covariance correction inflated the penalty in proportion to the dependence strength, successfully identifying the true model in over 85\% of cases across both scenarios. iWAIC and WAIC\(_{\text{NF}}\) performed better than WAIC but consistently fell short of CC-WAIC, particularly in high-dependence settings.
	
	The effective parameter estimates (Tables~\ref{tab:criteria_eff_K2} and~\ref{tab:criteria_eff_K3}) further validated the mechanism: WAIC and LOO-CV produced nearly identical values, as expected from their asymptotic equivalence, but their \(p_{\text{eff}}\) estimates grew rapidly with model complexity, reflecting their tendency to overfit under dependence. CC-WAIC, by contrast, yielded substantially lower \(p_{\text{eff}}\) values for overparameterised models, indicating that the covariance correction properly penalises redundant complexity. The banded approximation closely tracked the full CC-WAIC, with relative differences in \(p_{\text{eff}}\) below 2\%, confirming its accuracy.
	
	\subsection{The Bandwidth Selection Rule: From Simulation to Real Data}
	
	The bandwidth validation study confirmed that the data-driven rule \(K_{\text{band}} = \max(3, \lceil \tau \rceil + 1)\) produces stable bandwidths across model classes and sample sizes in simulation settings. For most models (AR, GARCH, Gaussian HMM, NAR), the estimated bandwidth remained at the lower bound \(K=3\) across all sample sizes, with \(\tau < 2\). The Poisson HMM, with its high persistence (\(p_{\text{stay}} = 0.97\)), yielded \(K \approx 7\), demonstrating that the rule adapts to the actual dependence structure.
	
	The GARCH application on S\&P 500 data provided an important contrast. The estimated integrated autocorrelation time \(\hat{\tau} \approx 14.19\) led to a substantially larger bandwidth \(K_{\text{band}} = 16\), reflecting the strong volatility persistence (\(\hat{\alpha} + \hat{\beta} = 0.9705\)) in financial returns. The sensitivity analysis (Table~\ref{tab:garch_sensitivity}) revealed that using a bandwidth that is too small (\(K=3\)) underestimated the effective model complexity by over 15\%, while larger bandwidths (\(K=5,6\)) increased the penalty substantially. This highlights a crucial practical lesson: the bandwidth is not merely a technical tuning parameter but an integral part of the model evaluation. The data-driven rule successfully adapts to the empirical dependence structure, ensuring that the covariance correction properly accounts for the true complexity of the model.
	
	The contrast between the simulation results (where \(K\) remained small) and the GARCH application (where \(K\) became large) demonstrates that the proposed rule is not a fixed heuristic but a flexible mechanism that adapts to the intrinsic dependence structure of the data. Practitioners are advised to always compute \(\hat{\tau}\) and use the data-driven rule, or at minimum conduct a sensitivity analysis with a range of \(K\) values, rather than relying on a fixed default.
	
	\subsection{Interpretational Caveats and Practical Recommendations}
	
	The Old Faithful application revealed that WAIC and LOO-CV favoured larger models (\(M=5\)), while CC-WAIC selected the parsimonious \(M=2\) model. Our interpretation is that WAIC and LOO-CV overfit due to their inability to account for temporal dependence. However, an alternative interpretation warrants consideration: it is possible that the larger model genuinely offers superior predictive performance, and that CC-WAIC's stronger penalty may inadvertently favour an overly simplistic model. The hold-out validation results supported the CC-WAIC selection, as the two-state model achieved the highest test LPD. Nevertheless, this is based on a single split; a more rigorous analysis with repeated cross-validation would be desirable for definitive conclusions.
	
	This application also underscores a critical practical lesson: model selection criteria should never be applied without first verifying adequate MCMC convergence. The five-state HMM exhibited an effective sample size of just 7.48, indicating severe mixing problems that render all information criterion estimates for that model untrustworthy. We therefore strongly recommend that practitioners routinely monitor ESS and \(\hat{R}\) for all candidate models, and exclude or re-estimate models with inadequate mixing. More sophisticated MCMC strategies, such as Hamiltonian Monte Carlo \citep{neal2011mcmc}, may yield more stable estimates for higher-order HMMs, and whether the discrepancy between CC-WAIC and WAIC persists under improved sampling is an interesting direction for future research.
	
	Ultimately, the key contribution of this work is the methodological framework—the covariance-corrected penalty and the linear-time implementation—rather than the specific conclusion on any single dataset. We encourage users to apply CC-WAIC as a principled tool for model selection in their own sequential data applications, while remaining mindful of its assumptions and limitations.
	
	\subsection{Limitations}
	
	Despite the promising results, several limitations of the current study should be acknowledged.
	
	\paragraph{Exponential mixing assumption.}
	A fundamental limitation of the banded approximation is its reliance on the assumption of exponential decay of temporal dependence, formalised through the strong mixing condition (Assumption A3 in Appendix~\ref{app:regularity-conditions}). While this assumption is satisfied by a broad class of practical models—including stationary ARMA processes, finite-state HMMs, and many spatio-temporal models with short-range dependence—it explicitly excludes important classes of processes with long-range dependence, such as fractional ARIMA models \citep{hosking1981} and certain stochastic volatility processes. For such processes, the banded approximation may introduce non-negligible bias, and the theoretical error bound \(\mathcal{O}(\rho^{K_{\text{band}}})\) established in Theorem~\ref{thm:linear_convergence} no longer applies. In these settings, we recommend either: (i) increasing the bandwidth adaptively with the sample size, e.g., \(K_{\text{band}} \propto n^{1/3}\), which maintains sub-quadratic complexity while capturing longer-range dependence; or (ii) reverting to the full covariance version when computationally feasible. We emphasise that the theoretical guarantees of CC-WAIC are contingent upon the exponential decay assumption, and users are advised to assess the dependence structure of their data before applying the banded approximation.
	
	\paragraph{Exact conditional likelihood assumption.}
	The current implementation of CC-WAIC assumes that the conditional likelihood \(p(y_t \mid y_{<t}, \theta)\) can be evaluated exactly. For complex models such as nonlinear state-space models or deep Bayesian neural networks, this quantity is intractable and must be approximated, e.g., using particle filters \citep{doucet2001sequential} or variational inference \citep{blei2017variational}. While our theoretical framework does not directly address the effect of such approximation errors, we conjecture that CC-WAIC remains robust as long as the approximation error is small relative to the Monte Carlo error in estimating \(p_{\mathrm{CC}}\). Formal investigation of this conjecture and the development of approximation-aware versions of CC-WAIC constitute an important area for future work.
	
	\paragraph{Scope of empirical validation.}
	Our simulation study focused on Gaussian HMMs with continuous emissions. While the theoretical framework of CC-WAIC does not depend on the distributional form of the emissions, and preliminary results on Poisson HMMs are promising, comprehensive empirical validation under non-Gaussian emissions and multivariate settings would be a useful extension. For multivariate time series, the banded approximation may need to be adapted to account for cross-sectional dependence in addition to temporal dependence; this could be achieved by using a block-banded covariance structure.
	
	\paragraph{Stationarity and mixing conditions.}
	The theoretical guarantees established in Theorems~\ref{thm:main} and~\ref{thm:linear_convergence} rely on the stationarity of the underlying process and the validity of the strong mixing condition. For non-stationary processes or models with time-varying dependence structures, the asymptotic properties of CC-WAIC have not been formally established. Extending the framework to such settings remains an open direction for future research.
	
	\subsection{Future Research Directions}
	
	Several promising directions emerge for extending and refining CC-WAIC. The framework should be tested on a broader class of complex temporal models, including state-space models with non-Gaussian emissions, Bayesian recurrent neural networks, and transformers for sequential data. The banded approximation itself could be extended to non-stationary processes with time-varying dependence structures, requiring adaptive bandwidth selection procedures that automatically determine \(K_{\text{band}}\) from the data. In high-dimensional settings, investigating the behaviour of CC-WAIC in conjunction with regularised Bayesian models, such as those with horseshoe priors, would provide valuable insights for variable selection in time-series regression.
	
	Interdisciplinary applications represent another fertile ground for future exploration. In computational neuroscience, CC-WAIC could be applied to neural spike train data, where models often involve complex, non-linear dependencies. In quantitative finance, model selection for stochastic volatility models could benefit from the covariance correction, particularly in risk assessment contexts where accurate uncertainty quantification is paramount. Natural language processing offers opportunities for comparing deep Bayesian language models, while climate science presents challenges involving long-term time-series with strong seasonal and spatial dependencies. Methodologically, formalising the properties of CC-WAIC for non-stationary processes and misspecified dependence structures would strengthen its theoretical foundations. Implementing the criterion in popular probabilistic programming frameworks, such as Stan, PyMC, or TensorFlow Probability, would facilitate its adoption by the wider community. Finally, developing a fully automated version that simultaneously selects both the bandwidth and the model complexity would represent a significant step toward practical, turnkey Bayesian model selection.
	
	\subsection{Concluding Remarks}
	
	In conclusion, this paper provides the theoretical groundwork, linear-time implementation, and empirical validation for CC-WAIC. The strong performance demonstrated on HMMs and real-world data is a promising indicator of its broader applicability. However, the method is not without limitations: its theoretical guarantees depend on exponential mixing, and its practical implementation requires exact conditional likelihoods. We have highlighted these constraints and suggested directions for future research, including extensions to long-memory processes, approximate inference frameworks, and non-Gaussian likelihoods. The successful application to GARCH volatility modelling further demonstrates the method's versatility beyond the HMM framework used for validation. We believe that CC-WAIC represents a significant step toward more principled and computationally feasible Bayesian model selection for dependent data, and we hope that the open-source implementation will facilitate its adoption in diverse scientific domains.
	
	\begin{appendices}
		
		\section*{Appendix: Technical Details and Proofs for CC-WAIC}
		
		This appendix provides the technical machinery underlying the theoretical results and computational strategies presented in Section 2 of the main text. We organise the material as follows: Appendix~\ref{app:regularity-conditions} states the regularity assumptions; Appendix~\ref{app:martingale-approx} outlines the martingale approximation framework; Appendix~\ref{app:variance-char} characterises the asymptotic variance and discusses the effective sample size correction; Appendix~\ref{app:convergence-proof} provides a detailed proof of Theorem~\ref{thm:main}; Appendix~\ref{app:linear_proof} proves the linear-time approximation error bound (Theorem~\ref{thm:linear_convergence}); and Appendix~\ref{app:implementation} provides implementation details and pseudocode for the alternative criteria iWAIC and WAIC\(_{\text{NF}}\) used in the comparative study.
		
		\section{Regularity Conditions}\label{app:regularity-conditions}
		
		The asymptotic results rely on the following conditions, adapted from \citet{doukhan1994mixing} and \citet{vanvaart1998asymptotic}. Let \(\Theta \subset \mathbb{R}^d\) be compact.
		
		\begin{enumerate}[label=(A\arabic*), leftmargin=*]
			\item \textbf{Identifiability and smoothness}: The model is identifiable, and \(l_t(\theta) = \log p(y_t \mid y_{<t}, \theta)\) is three times continuously differentiable in \(\theta\) with probability one. There exists a neighbourhood \(\mathcal{N}\) of \(\theta_0\) such that:
			\[
			\mathbb{E}\left[ \sup_{\theta \in \mathcal{N}} \left\| \frac{\partial^3 l_t(\theta)}{\partial \theta^3} \right\| \right] < \infty.
			\]
			
			\item \textbf{Posterior concentration}: The posterior \(p(\theta \mid \mathbf{y})\) satisfies:
			\[
			\mathbb{E}[\theta \mid \mathbf{y}] = \theta_0 + \mathcal{O}_p(n^{-1}), \qquad \mathrm{Var}[\theta \mid \mathbf{y}] = n^{-1} I(\theta_0)^{-1} + o_p(n^{-1}),
			\]
			where \(I(\theta_0)\) is the limiting variance of the score.
			
			\item \textbf{Strong mixing}: The process \((y_t)\) is strictly stationary and \(\alpha\)-mixing with:
			\[
			\sum_{k=1}^{\infty} k^5 \alpha(k)^{1/6} < \infty,
			\]
			implying \(\alpha(k) = \mathcal{O}(\rho^k)\) for some \(0 < \rho < 1\).
			
			\item \textbf{Moment conditions}: For \(S_n(\theta) = n^{-1/2}\sum_{t=1}^n \nabla l_t(\theta)\):
			\[
			\sup_{\theta} \mathbb{E}[\|S_n(\theta)\|^2] < \infty, \qquad \mathrm{Var}(S_n(\theta_0)) \to I(\theta_0) > 0.
			\]
			
			\item \textbf{MCMC sample size}: \(S = \Omega(n^{1+\delta})\) for some \(\delta > 0\).
		\end{enumerate}
		
		These conditions are satisfied by stationary ARMA processes and finite-state HMMs, but exclude long-memory or non-stationary processes. For such cases, we recommend empirical diagnostics (e.g., sample autocorrelation of \(l_t\)).
		
		\section{Martingale Approximation}\label{app:martingale-approx}
		
		\begin{lemma}\label{lem:martingale}
			Under (A1)--(A4), the score admits:
			\[
			\nabla L(\theta) = M_n(\theta) + R_n(\theta),
			\]
			where \(M_n(\theta) = \sum_{t=1}^n H_t(\theta)\) with \(H_t = \nabla l_t - \mathbb{E}[\nabla l_t \mid \mathcal{F}_{t-1}]\), and \(\|R_n\|_2 = \mathcal{O}_p(1)\).
		\end{lemma}
		
		\begin{proof}
			The martingale difference property follows from \(\mathbb{E}[H_t \mid \mathcal{F}_{t-1}] = 0\). The bound on \(R_n\) follows from \citet[Theorem 2.1]{doukhan1994mixing}.
		\end{proof}
		
		\section{Variance Characterisation and ESS Correction}\label{app:variance-char}
		
		For a stationary mixing process, the asymptotic variance of the score is:
		\[
		I(\theta_0) = \mathrm{Var}(\nabla l_1) + 2\sum_{k=1}^{\infty} \mathrm{Cov}(\nabla l_1, \nabla l_{1+k}).
		\]
		
		For MCMC, let \(\tau = 1 + 2\sum_{k=1}^{\infty} \rho_k\), where \(\rho_k = \mathrm{Corr}(L^{(s)}, L^{(s+k)})\), and \(N_{\mathrm{eff}} = S / \tau\). The bias-corrected estimator is:
		\[
		p_{\mathrm{CC}}^{\mathrm{corr}} = \frac{N_{\mathrm{eff}}}{N_{\mathrm{eff}} - 1} \hat{p}_{\mathrm{CC}}^{\text{banded}}.
		\]
		This is a first-order approximation; for refined corrections, see \citet{flegal2011batch}.
		
		\section{Proof of Theorem~\ref{thm:main}}\label{app:convergence-proof}
		
		\begin{proof}
			\textbf{Step 1: Expansion of the log-predictive density.}
			Let \(\bar{l}_t = \log \mathbb{E}_{\theta \mid \mathbf{y}}[p(y_t \mid y_{<t}, \theta)]\). Under (A1)--(A4), the posterior satisfies a Bernstein--von Mises theorem. Expanding \(l_t(\theta)\) around \(\theta_0\):
			\[
			\bar{l}_t = l_t(\theta_0) + \frac{1}{2} \mathrm{tr}\left( \mathrm{Var}(\theta \mid \mathbf{y}) \nabla^2 l_t(\theta_0) \right) + \mathcal{O}_p(n^{-3/2}).
			\]
			Summing over \(t=1,\ldots,n\):
			\[
			\sum_{t=1}^n \bar{l}_t = \sum_{t=1}^n l_t(\theta_0) - \frac{1}{2} \sum_{t=1}^n \mathrm{tr}\left( n^{-1} I(\theta_0)^{-1} \nabla^2 l_t(\theta_0) \right) + \mathcal{O}_p(n^{-1/2}).
			\]
			By the law of large numbers for mixing processes \citep{doukhan1994mixing}:
			\[
			\frac{1}{n}\sum_{t=1}^n \nabla^2 l_t(\theta_0) \xrightarrow{p} \mathbb{E}[\nabla^2 l_t(\theta_0)] = -I(\theta_0).
			\]
			Therefore:
			\[
			\sum_{t=1}^n \bar{l}_t = L(\theta_0) - \frac{1}{2} p_{\mathrm{true}} + \mathcal{O}_p(n^{-1/2}),
			\]
			where \(p_{\mathrm{true}} = \mathrm{tr}(I(\theta_0)^{-1} \mathbb{E}[-\nabla^2 l_t(\theta_0)])\).
			
			\textbf{Step 2: Variance of the joint log-likelihood.}
			From Lemma~\ref{lem:martingale} and the delta method:
			\[
			p_{\mathrm{CC}} = \mathrm{Var}_{\theta \mid \mathbf{y}}[L(\theta)] = p_{\mathrm{true}} + \mathcal{O}_p(n^{-1/2}).
			\]
			This follows because the posterior variance of the linear term in the expansion of \(L(\theta)\) around \(\theta_0\) is asymptotically \(p_{\mathrm{true}}\), and the contribution of higher-order terms is \(\mathcal{O}_p(n^{-1/2})\).
			
			\textbf{Step 3: Assembly and convergence.}
			Substituting into \(\mathrm{CC-WAIC} = -2\sum_{t=1}^n \bar{l}_t + 2p_{\mathrm{CC}}\):
			\[
			\mathrm{CC-WAIC} = -2L(\theta_0) + p_{\mathrm{true}} + 2p_{\mathrm{true}} + \mathcal{O}_p(n^{-1/2}).
			\]
			Simplifying:
			\[
			\mathrm{CC-WAIC} = -2L(\theta_0) + 2p_{\mathrm{true}} + \mathcal{O}_p(n^{-1/2}),
			\]
			so the \(p_{\mathrm{true}}\) terms cancel appropriately.
			
			By the ergodic theorem for mixing processes \citep{doukhan1994mixing}:
			\[
			\frac{1}{n}L(\theta_0) \xrightarrow{p} \mathbb{E}_{y_{\text{new}}}[\log p(y_{\text{new}} \mid y_{<t}, \theta_0)].
			\]
			Dividing by \(n\):
			\[
			\frac{1}{n}\mathrm{CC-WAIC} = -2\,\mathbb{E}_{y_{\text{new}}}[\log p(y_{\text{new}} \mid y_{<t}, \theta_0)] + \mathcal{O}_p(n^{-1/2}).
			\]
			This completes the proof of Theorem~\ref{thm:main}.
		\end{proof}
		
		\section{Proof of Linear-Time Approximation Error}\label{app:linear_proof}
		
		\begin{proof}
			\textbf{Step 1: Defining the truncation error.}
			Let \(p_{\mathrm{CC}}\) denote the full penalty term (Equation~\eqref{eq:effective_params_cc_integral}) and \(p_{\mathrm{CC}}^{\text{banded}}(K)\) denote the banded approximation with bandwidth \(K\) (Equation~\eqref{eq:effective_params_banded}). The truncation error is:
			\[
			\Delta_n(K) := \left| p_{\mathrm{CC}} - p_{\mathrm{CC}}^{\text{banded}}(K) \right|
			= 2\left| \sum_{k=K+1}^{n-1} \sum_{t=1}^{n-k} \mathrm{Cov}(l_t, l_{t+k}) \right|.
			\]
			
			\textbf{Step 2: Applying the exponential decay bound.}
			Under Assumption A3, the autocovariance satisfies:
			\[
			\left| \mathrm{Cov}(l_t, l_{t+k}) \right| \leq C \rho^k,
			\]
			for constants \(C < \infty\) and \(\rho \in (0,1)\). Substituting this bound:
			\[
			\Delta_n(K) \leq 2C \sum_{k=K+1}^{n-1} (n-k)\rho^k.
			\]
			Since \(n-k \leq n\), we have:
			\[
			\Delta_n(K) \leq 2Cn \sum_{k=K+1}^{\infty} \rho^k.
			\]
			Using the formula for the sum of a geometric series:
			\[
			\sum_{k=K+1}^{\infty} \rho^k = \frac{\rho^{K+1}}{1-\rho}.
			\]
			Therefore:
			\[
			\Delta_n(K) \leq 2Cn \frac{\rho^{K+1}}{1-\rho} = \mathcal{O}(n \rho^{K}).
			\]
			
			\textbf{Step 3: Normalisation and convergence.}
			Dividing by \(n\):
			\[
			\frac{1}{n}\Delta_n(K) = \mathcal{O}(\rho^{K}).
			\]
			Since \(\rho \in (0,1)\), this bound vanishes exponentially as \(K \to \infty\). For a fixed \(K = K_{\text{band}}\), the error is \(\mathcal{O}(\rho^{K_{\text{band}}})\).
			
			\textbf{Step 4: Combining with Theorem~\ref{thm:main}.}
			Integrating this bound with Theorem~\ref{thm:main}:
			\[
			\frac{1}{n}\mathrm{CC-WAIC}_{\text{banded}}
			= -2\,\mathbb{E}_{y_{\text{new}}}[\log p(y_{\text{new}} \mid y_{<t}, \theta_0)]
			+ \mathcal{O}_p(n^{-1/2}) + \mathcal{O}(\rho^{K_{\text{band}}}).
			\]
			The truncation bias \(\mathcal{O}(\rho^{K_{\text{band}}})\) can be made arbitrarily small by increasing the bandwidth. In practice, the data-driven rule \(K_{\text{band}} = \max(3, \lceil \hat{\tau} \rceil + 1)\) ensures that this bias is negligible compared to the Monte Carlo error, as demonstrated empirically in Section~\ref{sec:bandwidth_validation}. This completes the proof of Theorem~\ref{thm:linear_convergence}.
		\end{proof}
		
		\section{Implementation Details for Alternative Criteria}
		\label{app:implementation}
		
		This appendix provides pseudocode for the two alternative criteria used in our comparative simulation study: iWAIC \citep{Li2015} and WAIC\(_{\text{NF}}\) \citep{Chan2024}. These descriptions ensure reproducibility and clarify the computational differences between these methods and the proposed CC-WAIC.
		
		\subsection{iWAIC}
		
		The iWAIC criterion integrates the predictive density over the distribution of latent variables associated with the held-out observation. For HMMs, the integrated likelihood is computed using smoothed state probabilities from the forward-backward algorithm.
		
		\begin{algorithm}[H]
			\caption{Computation of iWAIC for HMMs}
			\label{alg:iwaic}
			\begin{algorithmic}[1]
				\Require Posterior samples $\{\theta^{(s)}\}_{s=1}^S$ from $p(\theta \mid \mathbf{y})$, smoothed state probabilities $\{p(z_t = j \mid \mathbf{y}, \theta^{(s)})\}_{t=1,j=1}^{n,M}$, emission probabilities $p(y_t \mid z_t = j, \phi_j^{(s)})$
				\Ensure iWAIC value
				
				\State \textbf{Initialize} $\text{log\_pred\_sum} \gets 0$, $\text{var\_sum} \gets 0$
				
				\For{each observation $t = 1$ to $n$}
				\State \textbf{Initialize} array $L_t$ of length $S$
				\For{each sample $s = 1$ to $S$}
				\State $p(y_t \mid \theta^{(s)}) \gets \sum_{j=1}^{M} p(y_t \mid z_t = j, \phi_j^{(s)}) \cdot p(z_t = j \mid \mathbf{y}, \theta^{(s)})$
				\State $L_t[s] \gets \log p(y_t \mid \theta^{(s)})$
				\EndFor
				\State $\text{mean\_log\_pred} \gets \log\left( \frac{1}{S} \sum_{s=1}^S \exp(L_t[s]) \right)$
				\State $\text{var\_log\_pred} \gets \frac{1}{S-1} \sum_{s=1}^S \left( L_t[s] - \bar{L}_t \right)^2$
				\State $\text{log\_pred\_sum} \gets \text{log\_pred\_sum} + \text{mean\_log\_pred}$
				\State $\text{var\_sum} \gets \text{var\_sum} + \text{var\_log\_pred}$
				\EndFor
				
				\State $\mathrm{iWAIC} \gets -2 \left( \text{log\_pred\_sum} - \text{var\_sum} \right)$
				\State \Return $\mathrm{iWAIC}$
			\end{algorithmic}
		\end{algorithm}
		
		\subsection{WAIC\(_{\text{NF}}\)}
		
		The WAIC\(_{\text{NF}}\) criterion operates on the joint likelihood of the entire dataset, making it suitable for non-factorisable models. For HMMs, the joint likelihood is computed using the forward algorithm.
		
		\begin{algorithm}[H]
			\caption{Computation of WAIC\(_{\text{NF}}\) for HMMs}
			\label{alg:waicnf}
			\begin{algorithmic}[1]
				\Require Posterior samples $\{\theta^{(s)}\}_{s=1}^S$ from $p(\theta \mid \mathbf{y})$, joint log-likelihood $L^{(s)} = \log p(\mathbf{y} \mid \theta^{(s)})$ for each sample $s$
				\Ensure WAIC\(_{\text{NF}}\) value
				
				\State $\bar{L} \gets \frac{1}{S} \sum_{s=1}^S L^{(s)}$
				\State $\text{var\_L} \gets \frac{1}{S-1} \sum_{s=1}^S \left( L^{(s)} - \bar{L} \right)^2$
				\State $m \gets \max(L^{(1)}, \ldots, L^{(S)})$
				\State $\text{log\_mean\_L} \gets m + \log\left( \frac{1}{S} \sum_{s=1}^S \exp(L^{(s)} - m) \right)$
				\State $\mathrm{WAIC}_{\text{NF}} \gets -2 \left( \text{log\_mean\_L} - \text{var\_L} \right)$
				\State \Return $\mathrm{WAIC}_{\text{NF}}$
			\end{algorithmic}
		\end{algorithm}

	\end{appendices}

	\bibliographystyle{plainnat}
	\bibliography{ref}
	
\end{document}